\documentclass[10pt,a4paper,twoside, reqno]{amsart}

\makeatletter

\def\subsection{\@startsection{subsection}{2}%
	\z@{.5\linespacing\@plus.5\linespacing}{.25\linespacing}%
	{\normalfont\bfseries}}

\def\subsubsection{\@startsection{subsubsection}{3}%
	\z@{.5\linespacing\@plus.5\linespacing}{.25\linespacing}%
	{\normalfont\itshape}}

\makeatletter
\def\@settitle{\begin{center}%
  \baselineskip14\p@\relax
	\normalfont\LARGE\scshape  
  \@title
  \end{center}%
}
\makeatother

 \usepackage{multirow}
\usepackage{color}
\usepackage[dvipsnames]{xcolor}

\usepackage[
bookmarks=true,         
unicode=false,          
pdftoolbar=true,        
pdfmenubar=true,        
pdffitwindow=false,     
pdfstartview={FitV},    
pdftitle={},    
pdfauthor={},     
pdfsubject={},   
pdfcreator={},   
pdfproducer={}, 
pdfkeywords={}, 
pdfnewwindow=true,      
raiselinks	= true,
colorlinks=true,       
linkcolor=MidnightBlue,          
citecolor=Mahogany,        
filecolor=ForestGreen,      
urlcolor=ForestGreen,           
plainpages=false,
pdfpagelabels
]{hyperref}

\date{\today}

\usepackage{eurosym}
\usepackage{amstext} 
\DeclareRobustCommand{\officialeuro}{%
	\ifmmode\expandafter\text\fi
	{\fontencoding{U}\fontfamily{eurosym}\selectfont e}}

\usepackage[T1]{fontenc}     
\usepackage[utf8]{inputenc}	
\usepackage{graphicx, keyval, trig, url} 

\usepackage{epstopdf}

\usepackage{subcaption}

\usepackage{textcomp}
\usepackage{booktabs}

\usepackage{gensymb}
\usepackage[capitalise]{cleveref}
\usepackage{bm}   

\usepackage{blindtext}

\usepackage{cite} 

\usepackage{dsfont,amssymb,amsmath, graphicx, amsthm} 
\usepackage{amsfonts,dsfont,mathtools, mathrsfs} 
\usepackage{algorithm,algorithmicx,fancyhdr}
\usepackage[noend]{algpseudocode}

\usepackage{calrsfs}
\DeclareMathAlphabet{\pazocal}{OMS}{zplm}{m}{n}
\let\mathcal\undefined
\DeclareMathAlphabet{\mathcal}{OMS}{cmsy}{m}{n}

\newtheorem{assumption}{Assumption}
\newtheorem{theorem}{Theorem}
\newtheorem{lemma}{Lemma}
\newtheorem{proposition}{Proposition}

\newtheorem{remark}{Remark}
\newtheorem{corollary}{Corollary}

\newif\ifarxiv 
\arxivtrue 

\ifarxiv
\title
[V.\ Rostampour]
{Stabilising Lifetime PD Models under Forecast Uncertainty}

\author[V.\ Rostampour]{Vahab Rostampour}

\thanks{The author is a Senior Quantitative Risk Analyst practitioner at top-tier banks, and a member of the IFRS Sustainability Reference Group (SRG).. {\tt https://vahabr.github.io/}
}

\date{\today}

\else

\fi

\begin{document} 

\ifarxiv

\begin{abstract}
Estimating lifetime probabilities of default (PDs) under IFRS~9 and CECL requires projecting point--in--time transition matrices over multiple years. 
A persistent weakness is that macroeconomic forecast errors compound across horizons, producing unstable and volatile PD term structures. 
This paper reformulates the problem in a state--space framework and shows that a direct Kalman filter leaves non--vanishing variability. 
We then introduce an anchored observation model, which incorporates a neutral long--run economic state into the filter. 
The resulting error dynamics exhibit asymptotic stochastic stability, ensuring convergence in probability of the lifetime PD term structure.
Simulation on a synthetic corporate portfolio confirms that anchoring reduces forecast noise and delivers smoother, more interpretable projections.
\end{abstract}

\maketitle


\section{Introduction}\label{introduction}

Forward--looking estimation of lifetime probabilities of default (PDs) is a core requirement of IFRS~9 \cite{IASB2014} and the Current Expected Credit Loss (CECL) framework \cite{FASB2016}. 
In practice, lifetime PDs are obtained by propagating a point--in--time (PIT) transition matrix over multiple quarterly horizons, starting from the current rating distribution \cite{Basel2000,Tasche2019}. 
Transition probabilities are conditioned on macroeconomic variables and calibrated to scenario forecasts provided by risk or economic research teams \cite{ECB2017}. 
Since forecasts are typically available only over limited horizons (e.g.\ five years), models assume a neutral long--run environment beyond this window and let the PD term structure converge towards through--the--cycle (TTC) behaviour \cite{EBA2017,BCBS2015}.

Macroeconomic drivers such as GDP growth or unemployment are inherently stochastic and subject to significant forecast error \cite{CroushoreStark2001}. 
Scenario projections provide plausible narratives but cannot anticipate shocks or structural breaks \cite{ReifschneiderTulip2017}. 
The COVID--19 crisis illustrates this challenge: forecasts projected a protracted downturn, while realised conditions recovered much faster \cite{OECD2020,IMF2021}. 
When such forecasts are fed into transition matrices, their errors embed directly in the PD propagation.  

The result is that long--horizon PDs inherit the full uncertainty of the scenario. 
Forecast errors are recycled quarter after quarter, producing volatility and drift in the term structure. 
This effect is most pronounced during stress, when forecasts deviate strongly from realised macro paths. 
The outcome is excessive volatility in lifetime PDs, complicating capital planning and undermining the interpretability of credit loss estimates \cite{EBA2020,BIS2021}.  

Several methods attempt to mitigate this issue by smoothing macro inputs. 
Scenario averaging \cite{Borio2014,EBA2017}, moving averages \cite{Orphanides2002,Hamilton1994}, and exponential smoothing \cite{Kreinin2001} reduce short--term noise but do not address the dynamic propagation of forecast error. 
In parallel, credit transition modelling has drawn on latent factor structures \cite{JarrowLandoTurnbull1997,DuffieSingleton1999,Lando2004} and Bayesian smoothing of transition matrices \cite{Kavvathas2001,Christensen2004}. 
These methods improve local estimation but do not resolve the long--horizon instability of PD propagation under forecast uncertainty.  

This paper makes two contributions. 
First, it establishes a rigorous instability result for multi--period propagation of PIT transition matrices under forecast uncertainty, showing that conventional schemes do not converge. 
Second, it introduces an anchored Kalman filter that incorporates a neutral long--run macro state into the update step. 
Anchoring prevents forecast error recycling and yields asymptotically stochastic stability: the estimation error remains mean--square bounded, and the PD term structure converges in probability to a meaningful TTC limit. 
The framework therefore links control--theoretic recursive estimation \cite{Kalman1960,Maybeck1979,Simon2006} with regulatory credit risk modelling under IFRS~9 and CECL.  

The remainder of the paper is organised as follows. 
Section~\ref{sec:problem} formulates lifetime PD estimation in a Markovian state--space framework and demonstrates instability under standard macro--driven propagation. 
Section~\ref{sec:formulation} introduces a Kalman filtering representation of the macroeconomic state and shows that naïve filtering does not eliminate error recycling. 
Section~\ref{sec:proposal} develops the anchored Kalman filter, derives its error dynamics, and proves asymptotic stochastic stability. 
Section~\ref{sec:simulation} reports simulation results on a synthetic corporate portfolio. 
Section~\ref{sec:final} concludes and outlines directions for future research.


\section{Problem Formulation}\label{sec:problem}

We now formulate the lifetime PD estimation problem as a discrete--time state--space system. 
The starting point is a TTC transition matrix describing long--run rating migrations. 
A macroeconomic driver is then introduced to generate PIT matrices that condition on forecasted scenarios. 
Propagating the rating distribution forward with these PIT matrices produces a 
lifetime PD term structure. 

While this setup reflects current industry and regulatory practice, it embeds macroeconomic forecast errors into every propagation step. 
Section~\ref{sec:instable} shows that such errors accumulate and prevent convergence, leading to unstable lifetime PD estimates. 
This motivates the filtering approach developed in subsequent sections.

\subsection{TTC Transition Matrix}\label{sec:ttcMatrix}

The TTC transition matrix summarises long--run average rating behaviour. 
It is typically estimated from historical observations using the cohort, or counting, method. 
This is the most widely used estimator in both academia and practice \cite{Nickell2000,LandoSkodeberg2002}. 
Alternative approaches include maximum likelihood and continuous--time Markov generators \cite{Kavvathas2001,Israel2001,Bluhm2016}, as well as Bayesian or smoothing techniques that provide confidence intervals and regularisation \cite{Christensen2004}. 
For the purpose of this paper we adopt the cohort method as baseline.

Formally, let $X_t \in \{1,\dots,K\}$ denote the discrete rating of an obligor at quarter $t$, where $K$ includes the default state. 
The TTC matrix $P^{\mathrm{TTC}} \in \mathbb{R}^{K \times K}$ is row--stochastic with entries
\[
p^{\mathrm{TTC}}_{ij}
=\frac{N_{ij}}{\sum_{j=1}^{K} N_{ij}},
\qquad
i,j \in \{1,\dots,K\},
\]
where $N_{ij}$ is the number of obligors that migrated from state $i$ to $j$ during the sample period. 
By construction, $p^{\mathrm{TTC}}_{ij} \ge 0$ and $\sum_{j=1}^{K} p^{\mathrm{TTC}}_{ij} = 1$ for all $i$.

\subsection{PIT Adjustment and Macroeconomic Driver}\label{sec:pitMatrix}

The TTC matrix reflects long--run average behaviour but does not incorporate current or expected macroeconomic conditions. 
To obtain a PIT transition matrix, a macroeconomic adjustment is applied to $P^{\mathrm{TTC}}$. 
This ``macro overlay'' links transition intensities to observable covariates such as GDP growth or unemployment \cite{Bangia2002,Pesaran2006,Duffie2007}. 
Regulatory guidance under Basel and IFRS~9 frameworks also recognises the need for PIT calibration and the use of macroeconomic adjustments \cite{BCBS2005,EBA2017}.

Formally, let $(M_t)_{t\ge0}$ denote the stochastic macroeconomic state (e.g.~GDP growth, unemployment) and $\widehat M_t$ its forecast at time $t$. 
The PIT transition matrix at quarter $t$ is defined by
\[
P_t \;=\; \mathcal{G}\!\left(P^{\mathrm{TTC}},\, \widehat M_t\right),
\]
where $\mathcal{G}: \mathbb{R}^{K \times K} \times \mathbb{R} \to \mathbb{R}^{K \times K}$ maps the TTC matrix and macro state into migration probabilities. 
A common specification is a logit--style overlay,
\[
p_{ij,t} \;\propto\; p^{\mathrm{TTC}}_{ij}\,\exp\!\big(\beta_{ij} \widehat M_t\big),
\qquad i,j \in \{1,\dots,K\},
\]
followed by row normalisation. 
The sensitivity coefficients $\beta_{ij}$ are estimated from historical rating transitions and macro variables, typically via multinomial logit or proportional hazard regressions \cite{Bangia2002,Pesaran2006}. 
Depending on data availability, the estimation may be conducted at the matrix level or directly on obligor--level events. 
The resulting coefficients quantify how strongly each migration probability responds to the macro factor and are used to adjust the TTC matrix forward under the forecasted macro path.

\subsection{Lifetime Propagation and State--Space Representation}\label{sec:LTPD}

Starting from an initial rating distribution $\pi_0 \in \mathbb{R}^{1\times K}$, the portfolio evolves quarter by quarter under the PIT transition matrices. 
At horizon $t+1$,
\[
\pi_{t+1} \;=\; \pi_t P_t,
\qquad t=0,1,\dots,T-1,
\]
where $P_t = \mathcal{G}\!\left(P^{\mathrm{TTC}},\, \widehat M_t\right)$ depends on the macroeconomic forecast $\widehat M_t$. 
The sequence $(\pi_t)_{t=0}^T$ is therefore a stochastic dynamical system driven by the forecasted macro path. 

Define $Y_t = \pi_t e_K$ as the probability of default by time $t$, where $e_K$ denotes the unit vector of the absorbing default state. 
The lifetime PD term structure is then given by the sequence $(Y_t)_{t=1}^T$, which inherits its dynamics from the entire forecast path $(\widehat M_s)_{s=0}^{t-1}$.


\subsection{Forecast Error and Instability of Lifetime Propagation}\label{sec:instable}

In practice, the true macroeconomic state $M_t$ is not observed; only its forecast $\widehat M_t$ is available. 
Define the forecast error
\[
\delta_t \;=\; \widehat M_t - M_t, \qquad t=0,1,\dots,T-1.
\]
Substituting $P_t=\mathcal{G}\!\big(P^{\mathrm{TTC}},\,M_t+\delta_t\big)$ into the rating propagation yields
\[
\pi_{t+1} \;=\; \pi_t \, \mathcal{G}\!\big(P^{\mathrm{TTC}},\,M_t+\delta_t\big).
\]

The macro process $(M_t)$ is often modelled as mean–reverting, consistent with standard autoregressive specifications in econometrics \cite{Hamilton1994,NelsonPlosser1982} and equilibrium term–structure models such as Vasicek \cite{Vasicek1977}. 
Even under such mean–reversion, the repeated injection of $\delta_t$ into $P_t$ can induce substantial variability in $(\pi_t)$ and, consequently, in the lifetime PD sequence $(Y_t)_t$. 
Because $\delta_t$ enters the transition matrix at every step, a non–vanishing error sequence $(\delta_t)$ accumulates over the lifetime horizon and prevents convergence to a stable limit.

We formalise this intuition in three parts: an elementary instability lemma, a rigorous non–convergence result, and a deviation bound that quantifies accumulation. 
Let
\[
\varphi(\pi,m) \;:=\; \pi \, \mathcal{G}\!\big(P^{\mathrm{TTC}},\,m\big)\, e_K,
\]
so that $Y_t=\varphi(\pi_t,M_t)$ denotes the error–free default probability at horizon $t$. 
Our reasoning is closely related to results on random products of stochastic matrices \cite{Seneta2006} and to the propagation of forecast error in macroeconomic time series \cite{Hamilton1994}. 
While sensitivity and non–convergence phenomena are known in stochastic matrix theory and stochastic approximation \cite{Seneta2006,Mitrophanov2005,KushnerYin2003}, to the best of our knowledge these instability properties have not been established for lifetime PD propagation under macroeconomic forecasting.

Appendix~\ref{app:proofs} presents the elementary instability result (Lemma~\ref{lem:intuitive}). 
The main implication for lifetime PDs is summarised next.

\begin{proposition}[Lifetime PD non--convergence under persistent forecast error]
\label{prop:nonconv}
Let $P^{\mathrm{TTC}}$ be row--stochastic and let $\mathcal{G}(\cdot,m)$ be row--stochastic and $C^1$ in $m$. 
Assume: (i) uniform, non--degenerate macro sensitivity of $Y_t=\varphi(\pi_t,m)$; 
(ii) i.i.d.\ forecast errors $\delta_t$ with $\mathbb{P}(|\delta_t|\ge \varepsilon)=p>0$; 
(iii) the error--free path $Y_t^\circ=\varphi(\pi_t^\circ,M_t)$ converges. 
Then $(Y_t)$ does not converge in probability.
\end{proposition}

\begin{proof}[Sketch]
By Borel--Cantelli, $|\delta_t|\ge \varepsilon$ occurs infinitely often with positive probability. 
Uniform sensitivity implies a fixed minimum deviation from the error--free sequence $Y_t^\circ$, so the distance to any candidate limit cannot vanish. 
Full details are given in Appendix~\ref{app:proofs}.
\end{proof}

A linear accumulation bound for deviations $e_t=\|\pi_t-\pi_t^\circ\|_1$ is provided in Corollary~\ref{cor:accum} (Appendix), which quantifies how forecast error adds up over time.

\begin{remark}[On assumptions (A1) and (A3)]
Assumption (A1) requires uniform non--degeneracy of the macro sensitivity of default probabilities. 
It holds, for example, if $\partial_m \mathcal G$ exists and $\inf_{\pi,m}\big|\partial_m\phi(\pi,m)\big|>0$ over the relevant domain. 
Assumption (A3) is standard when the error--free macro path is bounded or convergent and $\mathcal G$ induces an ergodic product of stochastic matrices. 
If $\delta_t\to 0$ almost surely, the non--convergence conclusion may fail. 
Section~\ref{sec:anchored_obs} addresses this by replacing raw forecasts with an anchored filter, which enforces $\delta_t\to 0$ and restores stability.
\end{remark}

The results in this subsection show that conventional macro--driven propagation does not yield a stable lifetime PD estimate: forecast errors enter the dynamics at each step, accumulate, and preclude convergence. 
To restore stability, a filtering mechanism is required that explicitly controls the impact of forecast uncertainty. 
Section~\ref{sec:formulation} therefore introduces a Kalman filtering representation of the macro state and examines whether such an approach can mitigate the instability identified here.


\section{Kalman Filtering for Lifetime PD Models}\label{sec:formulation}

This section develops a Kalman filtering framework for lifetime PD propagation. 
We first recall the essence of the Kalman filter as a recursive state estimator. 
We then discuss alternative design choices for applying the filter either to the macroeconomic driver or directly to the PD dynamics. 
A state--space representation of the macro driver is introduced, and the resulting na\"ive filter is analysed. 
The analysis shows that while Kalman filtering dampens short--run forecast noise, it does not resolve the structural instability identified in Section~\ref{sec:instable}.

\subsection{Essence of the Kalman Filter}

The Kalman filter is the canonical recursive estimator for linear stochastic state--space systems \cite{Kalman1960}. 
Consider latent state dynamics
\[
x_{t+1} = A x_t + w_t, 
\qquad w_t \sim \mathcal N(0,Q),
\]
and a measurement equation
\[
y_t = H x_t + v_t, 
\qquad v_t \sim \mathcal N(0,R),
\]
with $A$ the transition matrix, $H$ the observation matrix, and $Q,R$ the noise covariances. 
The filter recursively updates the posterior mean and covariance of $x_t$ given all past observations, alternating between prediction and update steps. 
Formally,
\[
\mu_t = \mu_{t|t-1} + K_t \big(y_t - H \mu_{t|t-1}\big),
\qquad
\Sigma_t = (I - K_t H)\Sigma_{t|t-1},
\]
where $K_t$ is the Kalman gain balancing prior uncertainty against measurement noise. 
The recursion ensures that forecast error is not amplified over time but is damped according to the reliability of new signals.

Since its introduction \cite{Kalman1960}, the filter has been applied widely in navigation, control, and econometrics \cite{Hamilton1994,KushnerYin2003}. 
In the present context, its relevance lies in treating macroeconomic forecasts as noisy signals of an unobserved ``true'' macro state. 
By filtering these signals, transition--matrix dynamics can be driven by smoothed state estimates rather than raw forecasts, a practice increasingly emphasised in regulatory discussions of forward--looking credit risk models \cite{EBA2017,BCBS2005}.


The Kalman filter can be incorporated into lifetime PD modelling in different ways. 
One option is to filter the macroeconomic state before it enters the transition matrix. 
Another is to filter the transition dynamics or PD sequence directly. 
Between these extremes, hybrid constructions are possible. 
This subsection compares the alternatives and explains the rationale for adopting macro--state filtering as baseline.

\paragraph{Filtering the Macroeconomic State.}
In this design the latent macro driver $M_t$ is modelled explicitly in a linear--Gaussian state--space form. 
The forecast $\widehat M_t$ is treated as a noisy observation,
\[
\begin{aligned}
\text{State:}& \quad M_{t+1} = A M_t + w_t, \quad w_t\sim\mathcal N(0,Q),\\
\text{Obs.:}& \quad \widehat M_t = H M_t + v_t, \quad v_t\sim\mathcal N(0,R),\\
\text{PD link:}& \quad P_t = \mathcal G\!\big(P^{\mathrm{TTC}},\mu_t\big),\quad 
\mu_t \;=\;\mathbb E[M_t|\widehat M_{0:t}].
\end{aligned}
\]
The PIT matrix is thus driven by the filtered estimate $\mu_t$ rather than the raw forecast.

\emph{Pros:}
\begin{itemize}
\item Preserves the TTC/absorbing structure of the transition matrix by construction. 
\item Conceptually transparent: the filter acts on macro noise, not on rating dynamics. 
\item Provides a natural entry point for introducing a neutral long--term anchor and proving stability. 
\end{itemize}
\emph{Cons:}
\begin{itemize}
\item Relies on correct specification of the macro--to--PD map $\mathcal G$. 
\item Requires calibration of $(A,Q,H,R)$, often difficult with limited macro--credit data. 
\end{itemize}

\paragraph{Filtering in Transition or PD Space.}
Alternatively, the filter may target the migration dynamics directly. 
Let $\theta_t$ denote latent parameters (e.g.\ logits or intensities) generating $P_t$:
\[
\begin{aligned}
\text{State:}& \quad \theta_{t+1}=\theta_t+\eta_t,\quad \eta_t\sim\mathcal N(0,\Sigma_\eta),\\
\text{Obs.:}& \quad C_{i\cdot,t}\sim\mathrm{Multinomial}(n_{i,t},P_{i\cdot,t}),\\
\text{Link:}& \quad P_{i\cdot,t}=\mathrm{softmax}(\theta_{i,t}), \ \text{with absorbing default enforced}.
\end{aligned}
\]

\emph{Pros:}
\begin{itemize}
\item Targets the transition dynamics directly, independent of $\mathcal G$. 
\item Can partially absorb misspecification in the macro overlay by learning from realised migrations. 
\end{itemize}
\emph{Cons:}
\begin{itemize}
\item Requires nonlinear filtering (EKF, UKF, particles), making analysis delicate. 
\item Data scarcity: migration counts are sparse and heteroskedastic, especially in low--default portfolios. 
\item Harder to justify for governance: filtering rating dynamics is less transparent than filtering macro inputs. 
\end{itemize}

\paragraph{Hybrid Variants.}
Between these designs several mid--ground options exist:
\begin{itemize}
\item \emph{Parameter--filter:} filter time--varying sensitivities $\beta_{ij,t}$ in 
$P_t=\mathcal G(P^{\mathrm{TTC}},\widehat M_t;\beta_t)$. 
\item \emph{Generator--filter:} work with a continuous--time generator $Q_t$ s.t.\ $P_t=\exp(Q_t\Delta)$, filtering the constrained parameters of $Q_t$. 
\item \emph{Augmented--state:} jointly filter $(M_t,\theta_t)$ to capture feedback between macro factors and transitions. 
\end{itemize}

\paragraph{Chosen Design.}
This paper adopts the \emph{macro--state filter}. 
It preserves the Markov structure, admits a neutral long--run anchor, and enables an explicit asymptotic stability result for the PD term structure, see Section~\ref{sec:macroSS}. 
Section~\ref{sec:naiveKF} shows that a direct Kalman application still leaves residual instability, while Section~\ref{sec:proposal} introduces the anchored formulation and establishes stability as the main contribution.

\subsection{State--Space Model for the Macroeconomic Driver}\label{sec:macroSS}

Let $M_t \in \mathbb{R}^r$ denote the latent macroeconomic state at quarter $t$ and $\widehat M_t \in \mathbb{R}^r$ the forecast available to the modeller. 
We assume a linear--Gaussian state--space model,
\[
M_{t+1} = A M_t + w_t, 
\qquad w_t \sim \mathcal N(0,Q),
\]
\[
\widehat M_t = H M_t + v_t, 
\qquad v_t \sim \mathcal N(0,R),
\]
with $A,H \in \mathbb{R}^{r\times r}$, process covariance $Q\succeq 0$, and observation covariance $R\succ 0$. 
The Kalman filter delivers the conditional mean $\mu_t=\mathbb E[M_t \mid \widehat M_{0:t}]$ and covariance $\Sigma_t=\mathrm{Cov}(M_t \mid \widehat M_{0:t})$ via the standard recursion:
\[
\begin{aligned}
\text{Predict:}\quad & \mu_{t|t-1} = A \mu_{t-1}, 
& \Sigma_{t|t-1} = A \Sigma_{t-1} A^\top + Q,\\[0.3em]
\text{Gain:}\quad & K_t = \Sigma_{t|t-1} H^\top \big(H \Sigma_{t|t-1} H^\top + R\big)^{-1},\\[0.3em]
\text{Update:}\quad & \mu_t = \mu_{t|t-1} + K_t\big(\widehat M_t - H \mu_{t|t-1}\big),\\
& \Sigma_t = (I - K_t H)\Sigma_{t|t-1}.
\end{aligned}
\]

The PIT transition matrix at time $t$ is then evaluated at the filtered macro state,
\[
P_t = \mathcal G\!\big(P^{\mathrm{TTC}},\, \mu_t\big),
\]
and the portfolio propagates according to $\pi_{t+1}=\pi_t P_t$. 
In this setup the Kalman filter acts as a smoothing device: rather than relying on raw forecasts $\widehat M_t$, the PD dynamics are driven by the filtered estimate $\mu_t$. 
The next part examines whether this ``na\"ive'' Kalman application is sufficient to stabilise lifetime PDs.

\subsection{Na\"ive Kalman Filtering and Residual Instability}\label{sec:naiveKF}

We first examine a direct application of the Kalman filter to macroeconomic forecasts. 
Here the forecast itself is treated as the observation of the latent state. 
Although this preserves the recursive update form, it means that each observation already embeds the same forecast error as the state being estimated. 
Forecast error is therefore recycled back into the filter, creating a feedback loop. 
Over long horizons, this error recycling prevents stochastic stability: estimation variance converges to a strictly positive constant and residual fluctuations remain in the PD term structure. 
This shortcoming motivates the anchored reformulation in Section~\ref{sec:proposal}, where the observation is tied to a neutral long--run macro state.

Formally, let the macro state follow
\[
M_{t+1} = A M_t + w_t,\qquad w_t \sim \mathcal N(0,Q),
\]
with forecast signal
\[
y_t = \widehat M_t = H M_t + v_t,\qquad v_t \sim \mathcal N(0,R),
\]
where $w_t$ and $v_t$ are independent, $Q\succeq 0$, $R\succ 0$. 
Let $\mu_t$ be the filtered mean, $\Sigma_t$ the filtered covariance, and $e_t=\mu_t-M_t$ the estimation error. 
The PIT matrix is evaluated at $\mu_t$, i.e.\ $P_t=\mathcal G(P^{\mathrm{TTC}},\mu_t)$, with portfolio dynamics $\pi_{t+1}=\pi_t P_t$.

Standard Kalman algebra gives the error recursion
\begin{equation}\label{eq:naive_error}
e_{t+1} = \big(I-K_{t+1}H\big)A e_t + \big(I-K_{t+1}H\big)w_t - K_{t+1}v_{t+1},
\end{equation}
with gain $K_{t+1}=\Sigma_{t+1|t}H^\top\big(H\Sigma_{t+1|t}H^\top+R\big)^{-1}$ and predictor covariance $\Sigma_{t+1|t}=A\Sigma_t A^\top+Q$. 
The covariance update is the Riccati recursion
\begin{equation}\label{eq:riccati}
\Sigma_{t+1} = \big(I-K_{t+1}H\big)\Sigma_{t+1|t}.
\end{equation}

\begin{theorem}[Residual variability under naive filtering]\label{thm:naiveKF}
Suppose $(A,Q^{1/2})$ is stabilisable, $(A,H)$ is detectable, $Q\succeq 0$, and $R\succ 0$. 
Then the Kalman error covariance converges to a positive--definite limit $\Sigma_\infty\succ 0$. 
If the PD map $\varphi(\pi,m)$ is locally Lipschitz and uniformly sensitive on a compact macro set, then
\[
\limsup_{t\to\infty}\,\mathbb E\big[\,|\,Y_t-\varphi(\pi_t,M_t)\,|\,\big] > 0,
\]
with $Y_t=\varphi(\pi_t,\mu_t)$. 
In particular, the lifetime PD sequence does not converge in probability.
\end{theorem}

\begin{proof}[Sketch]
Kalman theory implies $\Sigma_t\to\Sigma_\infty\succ 0$. 
The non--vanishing state error propagates into the PD sequence through Lipschitz continuity and sensitivity of $\varphi$, ensuring a strictly positive long--run deviation. 
Full details are provided in Appendix~\ref{app:proofs}.
\end{proof}

\begin{remark}[Self--observation pitfall]
If one replaces $y_t$ by $H\mu_{t-1}$, i.e.\ no new information, the innovation collapses to $H(I-A)\mu_{t-1}$ and the update becomes a deterministic linear map of $\mu_{t-1}$. 
This produces a biased estimator and breaks the filter’s error--correction logic. 
Such degenerate schemes are not considered further.
\end{remark}

\begin{figure}
  \centering
  \includegraphics[width=0.95\textwidth]{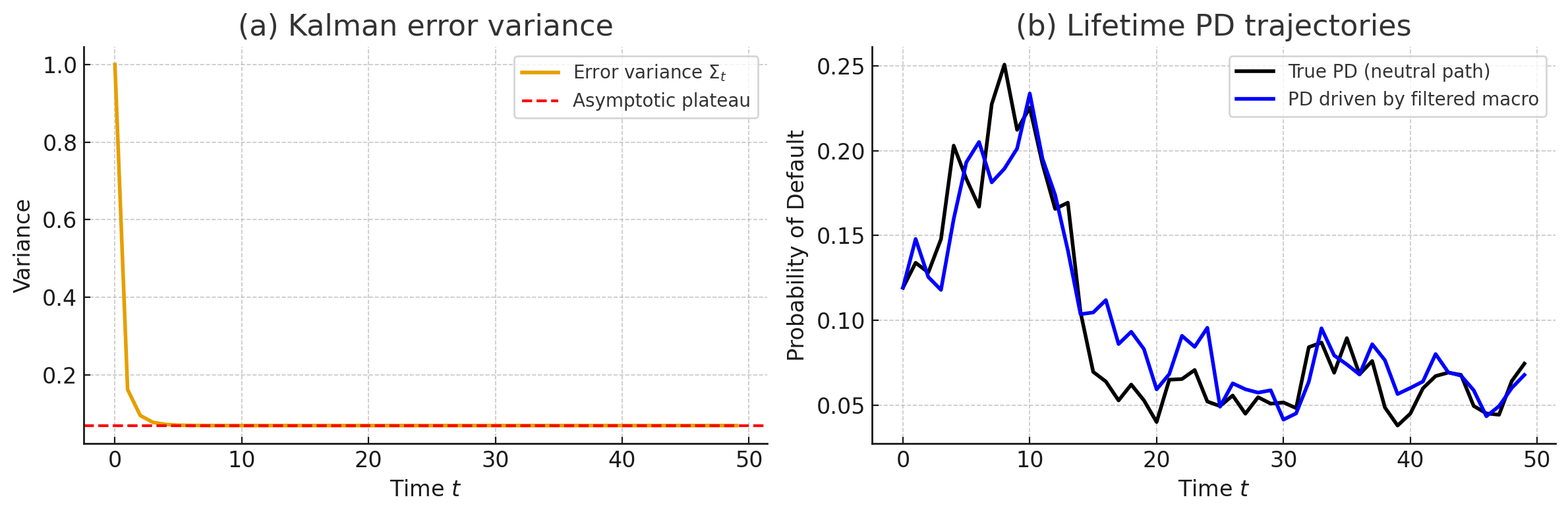}
  \caption{Residual variability under naive Kalman filtering. 
  (a) Estimation error variance $\mathbb E\|e_t\|^2$ converges to the positive constant 
  $\mathrm{tr}(\Sigma_\infty)$. 
  (b) Lifetime PD trajectories driven by the filtered macro state fluctuate indefinitely 
  around the neutral limit, preventing asymptotic stability. 
  Naive filtering prevents explosive growth of forecast errors but does not restore convergence.}
  \label{fig:naiveKF}
\end{figure}

\paragraph{Implication.}
Under the naive scheme, forecast noise survives the filter and perturbs the PIT matrices indefinitely. 
The lifetime PD term structure therefore fails to stabilise, which motivates the anchored observation design developed in Section~\ref{sec:proposal}.


\section{Anchored Filtering and Stability Result}\label{sec:proposal}

This section presents the proposed anchoring strategy for macroeconomic filtering. 
We first introduce the idea of neutral anchoring at a conceptual level, then derive a generalised anchored observation model. 
The resulting error dynamics are analysed and a rigorous proof of asymptotic stochastic stability is provided.

\subsection{Anchoring Strategies for Macroeconomic Filtering}

Section~\ref{sec:naiveKF} showed that a naive filter, which reuses forecasts as observations, 
recycles forecast error and leaves non--vanishing PD variability. 
To overcome this structural limitation, we introduce anchoring: replace the forecast proxy 
with a neutral macro state that serves as a stabilising reference. 
The idea is simple: each update is pulled towards a long--run neutral point, preventing 
forecast errors from compounding across horizons.

\paragraph{(i) Anchored macro--state filter (neutral proxy)}
\emph{Anchor:} fix a neutral macro state $M_\star$ (e.g.\ ``zero impact''). \\
\emph{Update:} $\mu_t=\mu_{t|t-1}+K_t\big(M_\star-H\mu_{t|t-1}\big)$. \\
\emph{Interpretation:} equivalent to observing $M_\star$ with vanishing noise; simple and stabilising. \\
\emph{Limitation:} calibration and governance are less transparent because the anchor enters implicitly.

\paragraph{(ii) Generalised anchored macro--state filter (stacked observation)}
\emph{Observation:} combine forecast and neutral anchor,
\[
y_t=\begin{bmatrix}\widehat M_t\\[0.2em] M_\star\end{bmatrix},\qquad 
H_{\mathrm{aug}}=\begin{bmatrix}H\\[0.2em] I\end{bmatrix},\qquad 
R_{\mathrm{aug}}=\mathrm{diag}\!\big(R,\ \sigma_\star^2 I\big).
\]
\emph{Update:} $\mu_t=\mu_{t|t-1}+K_t\big(y_t-H_{\mathrm{aug}}\mu_{t|t-1}\big)$. \\
\emph{Control knob:} $\sigma_\star^2$ tunes anchor strength; take $\sigma_\star^2>0$ within the forecast window and $\sigma_\star^2\downarrow 0$ beyond $T_F$. \\
\emph{Benefit:} explicit measurement structure, observability via the identity block, and a clean Riccati convergence path for stability proofs.

\paragraph{General formulation.}
The neutral proxy in (i) is contained in the stacked model (ii) as the limiting case $\sigma_\star^2\downarrow 0$. 
We therefore analyse the general anchored formulation going forward; all stability results apply to the proxy as a special case. 
Section~\ref{sec:stabilityproof} establishes asymptotic stochastic stability of the lifetime PD term structure under this model.

\subsection{Generalised Anchored Observation Model}\label{sec:anchored_obs}

To overcome the instability of the naive formulation, we introduce an anchored observation model. 
The key idea is to represent the macroeconomic state as a deviation from a neutral long--run level $M^\ast$, 
interpreted as the TTC equilibrium. 
Rather than updating relative to past forecasts, the filter evaluates each observation against this anchor, 
eliminating the accumulation of forecast error across horizons.

Formally, the observation equation is
\begin{equation}
y_t \;=\; H(M_t - M^\ast) + \nu_t, 
\qquad \nu_t \sim \mathcal N(0,R),
\end{equation}
where $H$ maps deviations from the anchor into the observable space. 
The state dynamics remain
\begin{equation}
M_{t+1} = A M_t + w_t, 
\qquad w_t \sim \mathcal N(0,Q),
\end{equation}
with $A$ capturing persistence and $Q$ the innovation covariance. 
Defining deviations $\delta_t := M_t - M^\ast$, we obtain
\begin{equation}
\delta_{t+1} = A \delta_t + w_t.
\end{equation}

The Kalman filter applied to this deviation system yields
\begin{align}
\hat{\delta}_{t|t-1} &= A \hat{\delta}_{t-1|t-1}, \\
\hat{\delta}_{t|t}   &= \hat{\delta}_{t|t-1} + K_t \big(y_t - H \hat{\delta}_{t|t-1}\big),
\end{align}
with Kalman gain $K_t = P_{t|t-1} H^\top (H P_{t|t-1} H^\top + R)^{-1}$ and covariance recursion as usual. 

By construction, $\delta_t$ is centred at zero under the neutral macro environment. 
The filter is therefore mean--reverting around the anchor and prevents divergence of lifetime PD forecasts induced by persistent forecast errors. 
This anchored formulation provides the foundation for the stability analysis in Section~\ref{sec:stabilityproof}, 
where asymptotic stochastic stability of the lifetime PD term structure is established.

\subsection{Error Dynamics and Convergence Intuition}\label{sec:errordyn}

To avoid confusion with the forecast error $\delta_t$ from Section~\ref{sec:instable}, 
we denote by 
\[
\xi_t := M_t - M^\star
\]
the deviation of the macro state from the neutral anchor $M^\star$. 
Under the anchored observation model,
\[
y_t = H \xi_t + \nu_t,\qquad \nu_t\sim\mathcal N(0,R),
\]
with state dynamics
\[
\xi_{t+1} = A \xi_t + w_t,\qquad w_t\sim\mathcal N(0,Q).
\]
Let $\hat\xi_{t|t}$ be the Kalman estimate and $e_t:=\hat\xi_{t|t}-\xi_t$ the estimation error.

\paragraph{Error recursion.}
The Kalman update yields
\[
\hat\xi_{t|t-1}=A\hat\xi_{t-1|t-1},\qquad
\hat\xi_{t|t}=\hat\xi_{t|t-1}+K_t(y_t-H\hat\xi_{t|t-1}),
\]
leading to the error dynamic
\begin{equation}\label{eq:err_deviation}
e_t = (I-K_t H)A e_{t-1} - (I-K_t H)w_{t-1} + K_t \nu_t .
\end{equation}
Thus the mean error evolves under the linear map $(I-K_t H)A$, with process and measurement noise entering additively. 

\paragraph{Stacked formulation.}
If both the forecast $\widehat M_t$ and the neutral anchor are observed, the system can be written in stacked form with 
\[
y_t^{\mathrm{aug}}=\begin{bmatrix}\widehat M_t\\ M^\star\end{bmatrix},\qquad
H_{\mathrm{aug}}=\begin{bmatrix}H\\ I\end{bmatrix},\qquad
R_{\mathrm{aug}}=\mathrm{diag}(R,\sigma_\star^2 I).
\]
The associated error recursion is
\begin{equation}\label{eq:err_stacked}
e_{t+1}^{\mathrm{aug}}
= (I-K_{t+1}H_{\mathrm{aug}})A e_t^{\mathrm{aug}}
+ (I-K_{t+1}H_{\mathrm{aug}})w_t - K_{t+1}\eta_{t+1},
\end{equation}
with $\eta_t=(v_t,\nu_t)^\top$. 
The identity block in $H_{\mathrm{aug}}$ guarantees detectability and, once the gain stabilises, ensures a strict contraction.

\paragraph{Convergence intuition.}
Anchoring changes the innovation reference: instead of tracking past forecasts, each update is pulled back towards $M^\star$. 
Large deviations therefore decay over time. 
The contraction $(I-K_t H)A$ (or $(I-K_t H_{\mathrm{aug}})A$ in the stacked model) implies that the estimation error remains mean–square bounded, and vanishes in the absence of process noise with a tight anchor. 
By Lipschitz continuity of the PD map $\phi$, the lifetime PD sequence $Y_t=\phi(\pi_t,\hat M_{t|t})$ inherits this stability. 
The next section formalises these statements.

\subsection{Proof of Asymptotic Stochastic Stability}\label{sec:stabilityproof}

We analyse the generalised anchored filter with stacked observation. 
The latent state is $M_t$, the observation is 
\[
y_t^{\mathrm{aug}}=\begin{bmatrix}\widehat M_t\\ M^\star\end{bmatrix},\qquad 
H_{\mathrm{aug}}=\begin{bmatrix}H\\ I\end{bmatrix},\qquad 
R_{\mathrm{aug}}=\mathrm{diag}(R,\sigma_\star^2 I).
\]
The portfolio evolves via $\pi_{t+1}=\pi_t P_t$ with $P_t=\mathcal G(P^{\mathrm{TTC}},\hat M_{t|t})$, 
and lifetime PDs are given by $Y_t=\phi(\pi_t,\hat M_{t|t})$.

\begin{assumption}\label{ass:KF}
$A\in\mathbb R^{r\times r}$, $Q\succeq 0$, $R\succ 0$, $\sigma_\star^2>0$. 
The pair $(A,Q^{1/2})$ is stabilisable and $(A,H_{\mathrm{aug}})$ is detectable. 
The PD map $\phi(\pi,m)$ is locally Lipschitz in $m$ near $M^\star$ with constant $L>0$.
\end{assumption}

\begin{theorem}[Anchored mean--square stability and convergence]\label{thm:MSS}
Under Assumption~\ref{ass:KF}, the Kalman error covariance $\Sigma_t$ converges to the unique 
stabilising solution $\Sigma_\infty$ of the Riccati equation for $(A,H_{\mathrm{aug}},Q,R_{\mathrm{aug}})$. 
In particular,
\[
\lim_{t\to\infty}\mathbb E\|e_t\|^2 \;=\; \mathrm{tr}\,\Sigma_\infty \;<\;\infty,
\]
so the estimation error is mean--square bounded and does not recycle forecast noise.

If, in addition, there exists a horizon $T_F$ such that $Q=0$ and $\sigma_\star^2=0$ for all $t\ge T_F$ 
(hard neutral anchor beyond the forecast window), then $e_t\to 0$ exponentially fast. 
Consequently,
\[
\phi(\pi_t,\hat M_{t|t}) \;\xrightarrow{\;p\;}\; \phi(\pi_t,M^\star),
\]
and the lifetime PD term structure converges in probability to the neutral macro limit.
\end{theorem}

\begin{proof}[Sketch]
Detectability of $(A,H_{\mathrm{aug}})$ with the identity block guarantees a stabilising Kalman gain, 
so $\Sigma_t\to\Sigma_\infty$ and the error is mean--square bounded. 
Under hard anchoring, $(I-KH_{\mathrm{aug}})A$ is a strict contraction, yielding exponential decay of $e_t$. 
By Lipschitz continuity of $\phi$, convergence of $e_t$ implies convergence of PDs. 
A complete proof is given in Appendix~\ref{app:proofs}.
\end{proof}

\paragraph{Technical notes.} 
The assumptions underlying Theorem~\ref{thm:MSS} and their practical justification 
are discussed further in the Appendix~\ref{app:proofs}. 
In particular, we comment on uniform sensitivity of logit overlays, compactness of the macro domain, 
and continuous anchoring within the forecast window. 
These notes clarify the robustness of the stability result while keeping the main exposition focused.


\section{Implementation and Simulation Results}\label{sec:simulation}

We illustrate the behaviour of lifetime PD propagation under synthetic scenarios. 
A minimal corporate portfolio is simulated over a five--year forecast horizon ($T_F=20$ quarters), 
followed by reversion to TTC dynamics. 
Three procedures are compared on identical data: 
(i) direct propagation without filtering, 
(ii) naive Kalman filtering using forecasts as observations, and 
(iii) the proposed anchored observation model. 
Key outputs are lifetime PD term structures, cross--scenario variance, and stability diagnostics.

\subsection{Synthetic Portfolio and Transition System}

We consider a portfolio of $N=10{,}000$ obligors, initially distributed as
\[
\pi_0 = (0.45,\;0.40,\;0.15,\;0.00), 
\qquad \text{counts } = (4500,\,4000,\,1500,\,0),
\]
across ratings $\{A,B,C,D\}$ with $D$ absorbing (default). 
The quarterly TTC matrix is
\[
P^{\mathrm{TTC}} =
\begin{bmatrix}
0.975 & 0.022 & 0.002 & 0.001 \\
0.030 & 0.935 & 0.030 & 0.005 \\
0.010 & 0.060 & 0.915 & 0.015 \\
0     & 0     & 0     & 1
\end{bmatrix}.
\]

The macro driver is a composite index
\[
M_t = \tfrac{1}{2}z(g_t)-\tfrac{1}{2}z(u_t),
\]
where $g_t$ is GDP growth, $u_t$ the unemployment rate, and $z(\cdot)$ denotes standardisation. 
The PIT matrix is constructed via a logit overlay
\[
p_{ij,t} \;\propto\; p^{\mathrm{TTC}}_{ij}\,\exp(\beta_{ij}M_t),
\qquad \sum_j p_{ij,t}=1,\; p_{ij,t}\ge0,
\]
with sensitivities on nearest--neighbour moves and default, e.g.\ $\beta_{A\to B}=2.0$, $\beta_{A\to D}=3.0$, $\beta_{B\to D}=2.0$, $\beta_{C\to D}=1.2$, and symmetric upgrades.

\subsection{Macroeconomic Scenarios}

Three stylised scenarios are simulated:

\begin{itemize}
  \item \textbf{Baseline:} mild cycle, GDP at $0.5\%$ per quarter with small oscillations; unemployment flat near $5.5\%$. Realisations add Gaussian noise ($\sigma_g=0.2\%$, $\sigma_u=0.2$pp).  
  \item \textbf{Stress:} downturn with slow recovery, mimicking a financial crisis. GDP falls $-2.0\%$ in $t=1$, recovers linearly; unemployment rises to $7.5\%$ then recedes. Realised path: deeper trough and slightly higher unemployment peak.  
  \item \textbf{Pandemic:} abrupt contraction with rapid rebound. GDP $-8\%$ at $t=2$, $+6\%$ at $t=3$, then normalisation; unemployment spikes to $9.5\%$ then falls quickly. Realised path: faster rebound imposed at $t=3$.  
\end{itemize}

Figures~\ref{fig:macro_scenarios}--\ref{fig:macro_params} display the forecast and realised macro paths. 
The baseline remains close to equilibrium; the stress scenario shows persistent weakness; 
the pandemic is characterised by a short, sharp disruption.

\subsection{Filtering Designs}

The macro state follows an AR(1) process
\[
M_{t+1}=\rho M_t+w_t,\quad w_t\sim\mathcal N(0,Q),
\qquad \rho=0.90,\; Q=0.19,
\]
with forecast $y_t$ as a noisy signal.

\begin{itemize}
\item \textbf{Naive KF:} observation $y_t=\widehat M_t$, with variance $R=0.25$.  
\item \textbf{Anchored KF:} stacked observation $y_t^{\mathrm{aug}}=(\widehat M_t,M^\star)$, 
$R_{\mathrm{aug}}=\mathrm{diag}(0.25,\sigma_\star^2)$, 
with $\sigma_\star^2=0.25$ inside the horizon and $\sigma_\star^2=0$ afterwards.  
\end{itemize}

This design ensures responsiveness to PIT signals during the forecast window while enforcing TTC convergence beyond.

\paragraph{Parameter summary.}
\[
\begin{array}{lcl}
\text{Persistence} & \rho &= 0.90 \\
\text{Process var.} & Q    &= 0.19 \\
\text{Forecast var.} & R   &= 0.25 \\
\text{Anchor var. (in)} & \sigma_\star^2 &= 0.25 \\
\text{Anchor var. (out)} & \sigma_\star^2 &= 0 \\
\text{Horizon} & T_F &= 20 \\
\text{Portfolio size} & N &= 10{,}000
\end{array}
\]

\paragraph{Robustness.}
Results are qualitatively unchanged for $\rho\in[0.8,0.95]$. 
Higher $R$ increases dispersion under the naive KF but does not destabilise the anchored filter. 
Smaller $\sigma_\star^2$ accelerates convergence to $M^\star$ but reduces PIT sensitivity; 
a decaying schedule $\sigma_{\star,t}^2\downarrow0$ is a practical compromise.

\begin{figure}[htbp]
  \centering
  \begin{subfigure}[t]{0.32\textwidth}
    \centering
    \includegraphics[width=0.9\textwidth]{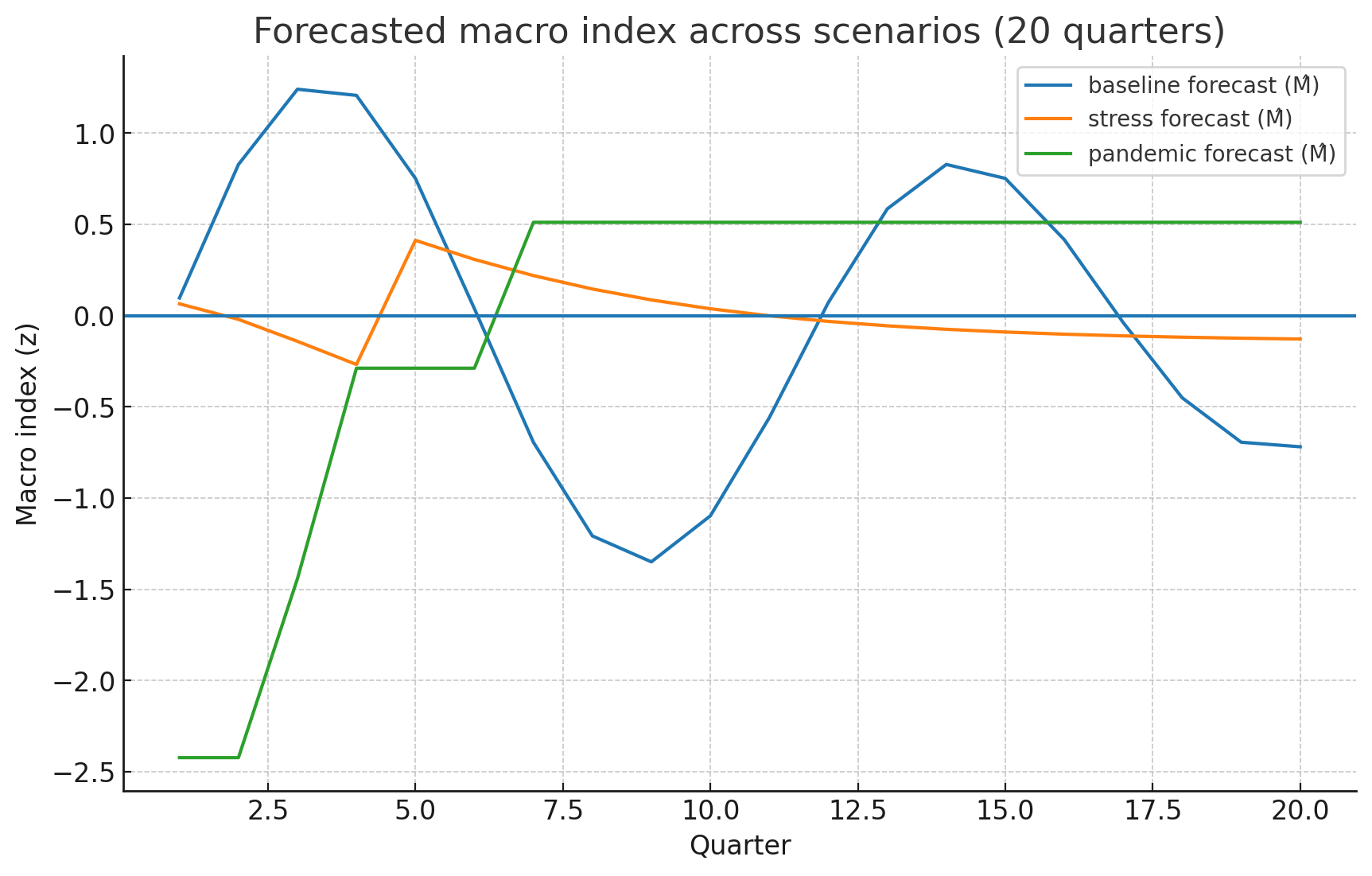}
    \caption{Composite macro index $M_t$ under baseline (mild cycle), stress (sharp downturn, slow recovery), and pandemic (abrupt shock with rebound). The neutral line at zero indicates TTC conditions.}
    \label{fig:macro_index}
  \end{subfigure}
  \begin{subfigure}[t]{0.32\textwidth}
    \centering
    \includegraphics[width=\textwidth]{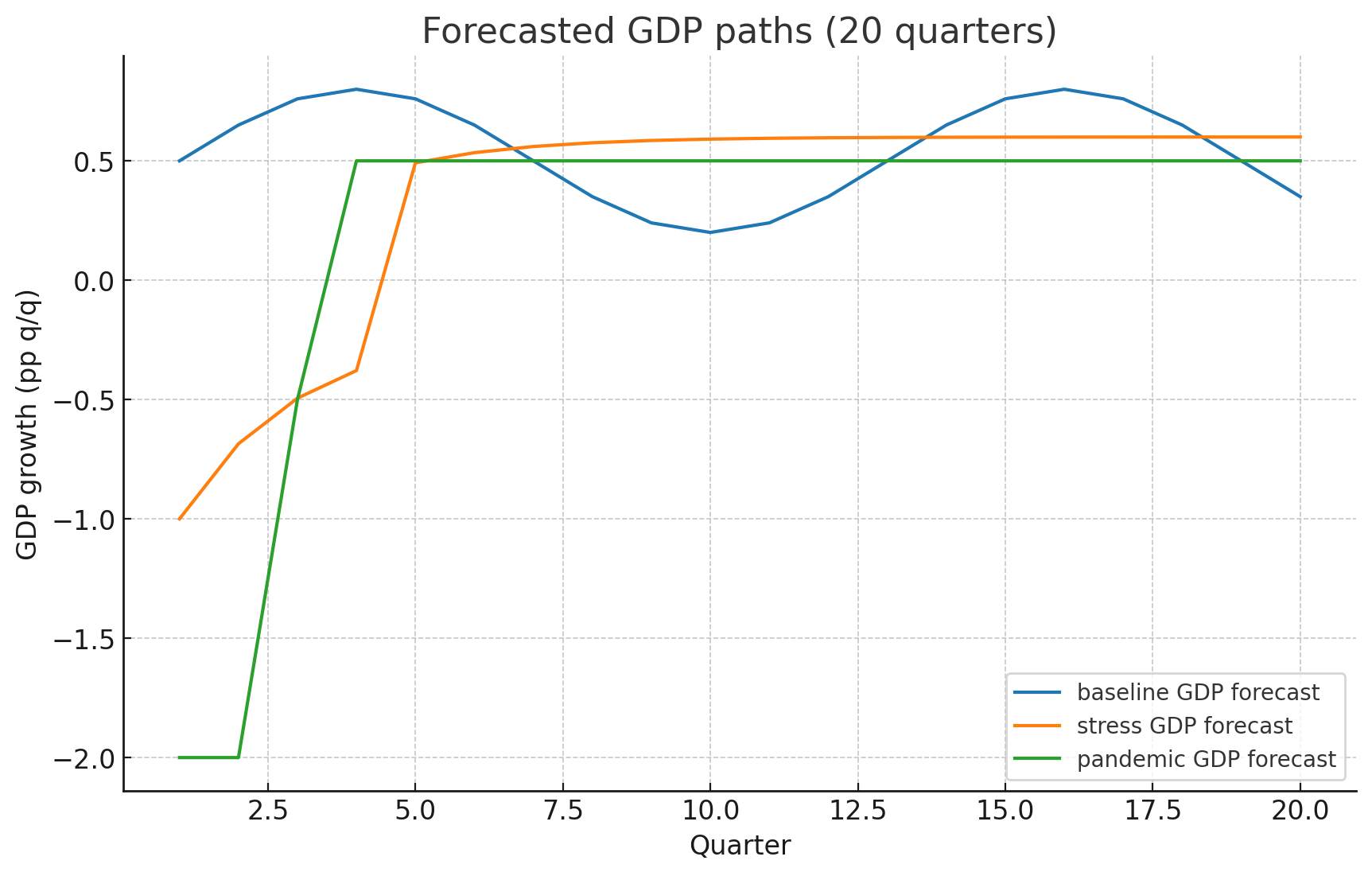}
    \caption{Quarterly GDP growth forecasts. Baseline oscillates around $0.5\%$, stress contracts then recovers gradually, pandemic exhibits a sharp fall followed by rebound.}
    \label{fig:gdp_forecast}
  \end{subfigure}
  \begin{subfigure}[t]{0.32\textwidth}
    \centering
    \includegraphics[width=\textwidth]{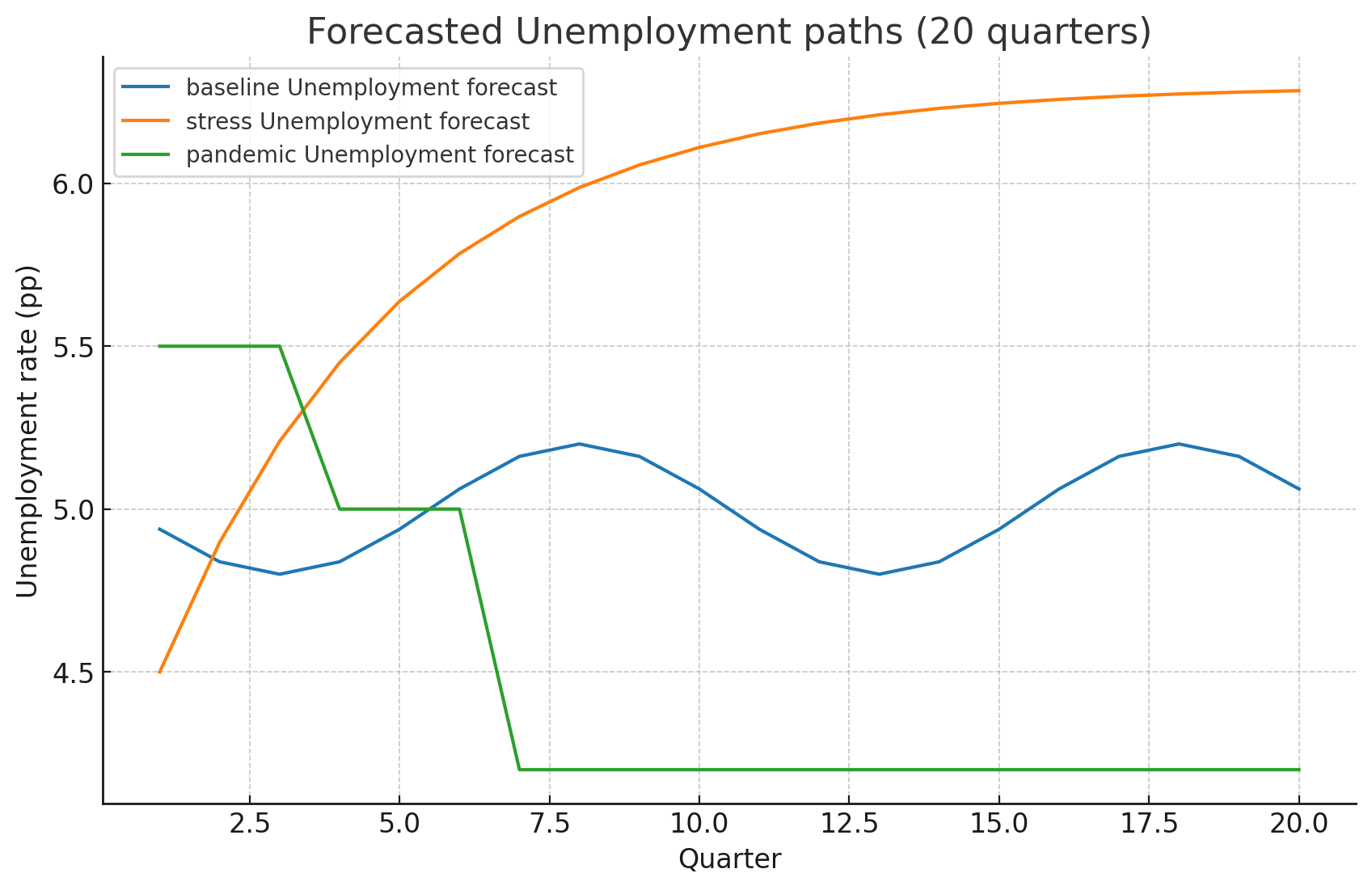}
    \caption{Unemployment forecasts. Baseline remains near $5\%$, stress rises to a persistent plateau, pandemic spikes abruptly then normalises.}
    \label{fig:unemp_forecast}
  \end{subfigure}

  \caption{Forecasted macroeconomic scenarios (baseline, stress, pandemic) over a 20--quarter horizon: (a) composite index $M_t$, (b) GDP growth, and (c) unemployment rate.}
  \label{fig:macro_scenarios}
\end{figure}

\begin{figure}[htbp]
  \centering

  \begin{subfigure}[t]{0.32\textwidth}
    \centering
    \includegraphics[width=\textwidth]{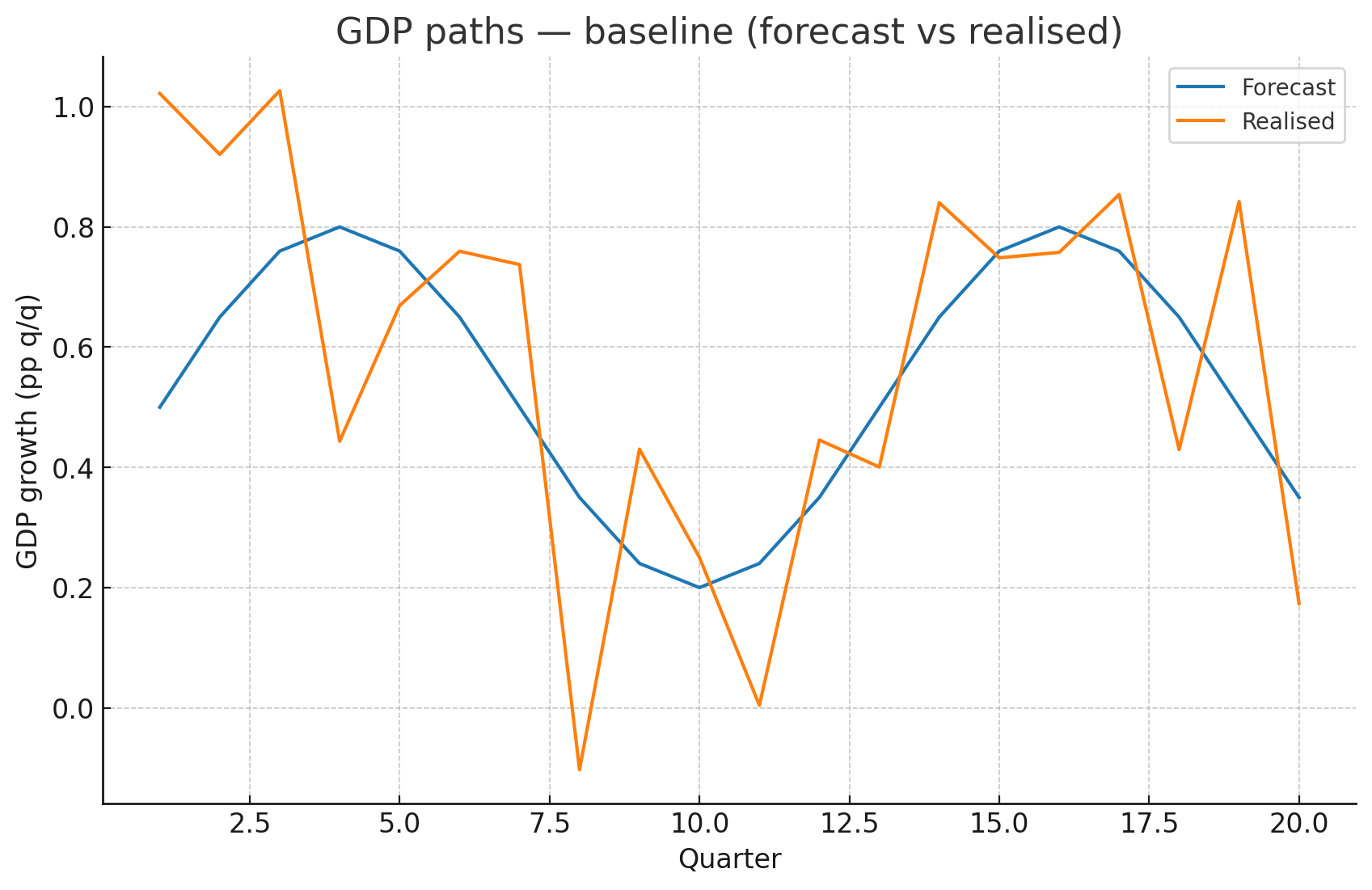}
    \caption{GDP growth, baseline scenario. Forecast (blue line) oscillates near trend; realised path (black dots) fluctuates around it.}
    \label{fig:gdp_baseline}
  \end{subfigure}
  \begin{subfigure}[t]{0.32\textwidth}
    \centering
    \includegraphics[width=\textwidth]{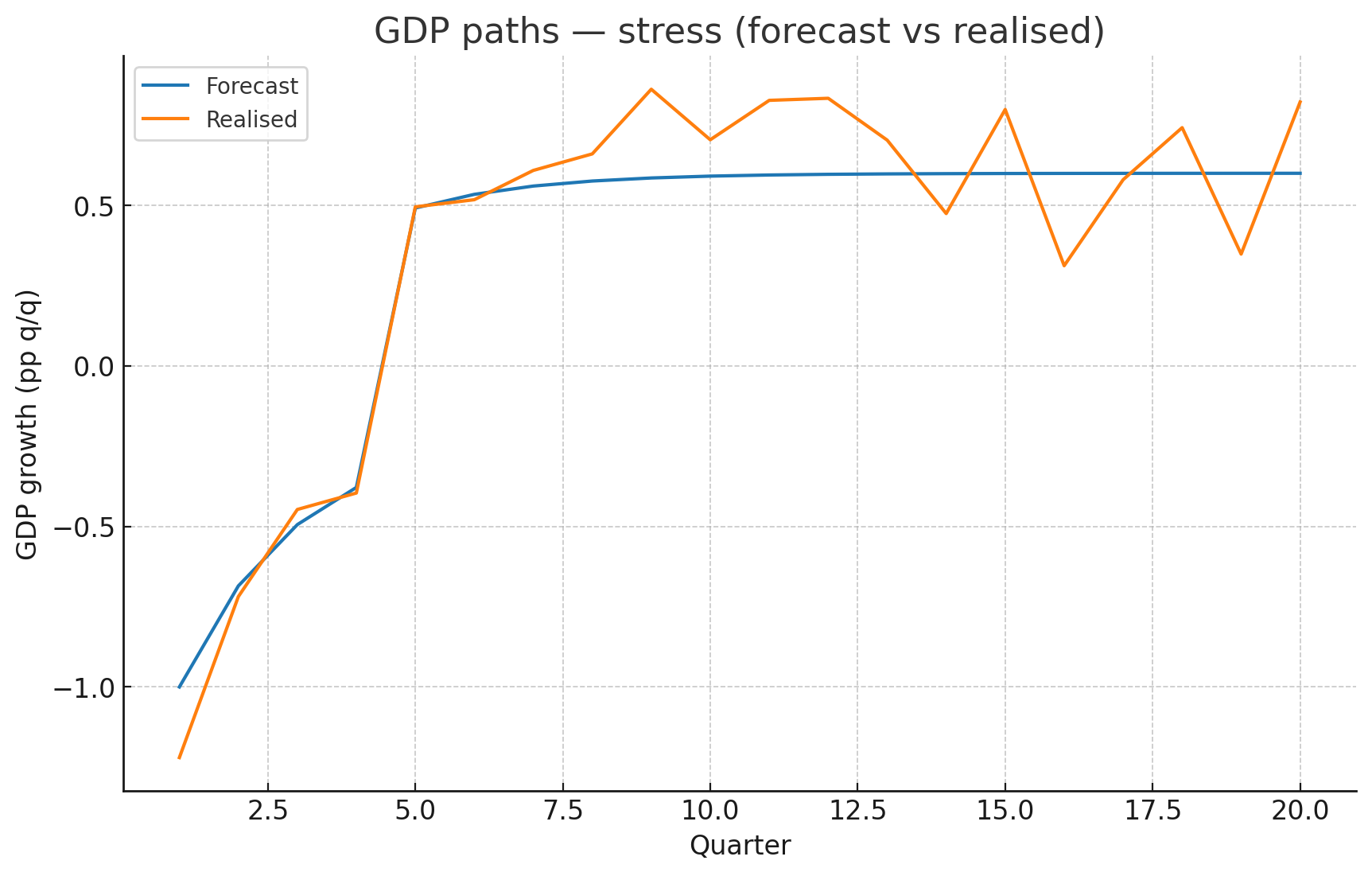}
    \caption{GDP growth, stress scenario. Forecast (orange) shows contraction then gradual recovery; realised path deviates with deeper trough.}
    \label{fig:gdp_stress}
  \end{subfigure}
  \begin{subfigure}[t]{0.32\textwidth}
    \centering
    \includegraphics[width=\textwidth]{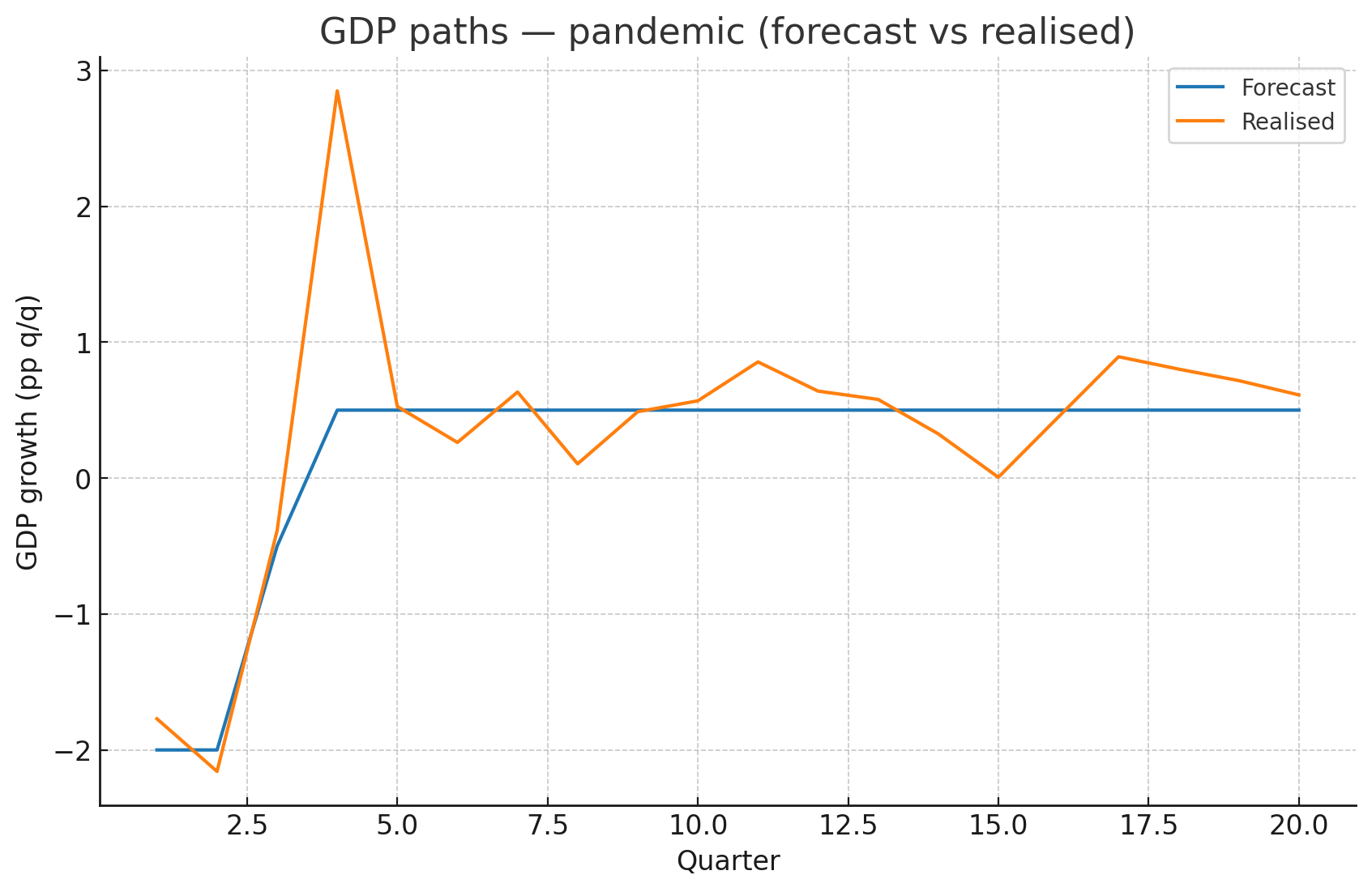}
    \caption{GDP growth, pandemic scenario. Forecast (green) imposes sharp fall and rebound; realised path overshoots during recovery.}
    \label{fig:gdp_pandemic}
  \end{subfigure}

  \begin{subfigure}[t]{0.32\textwidth}
    \centering
    \includegraphics[width=\textwidth]{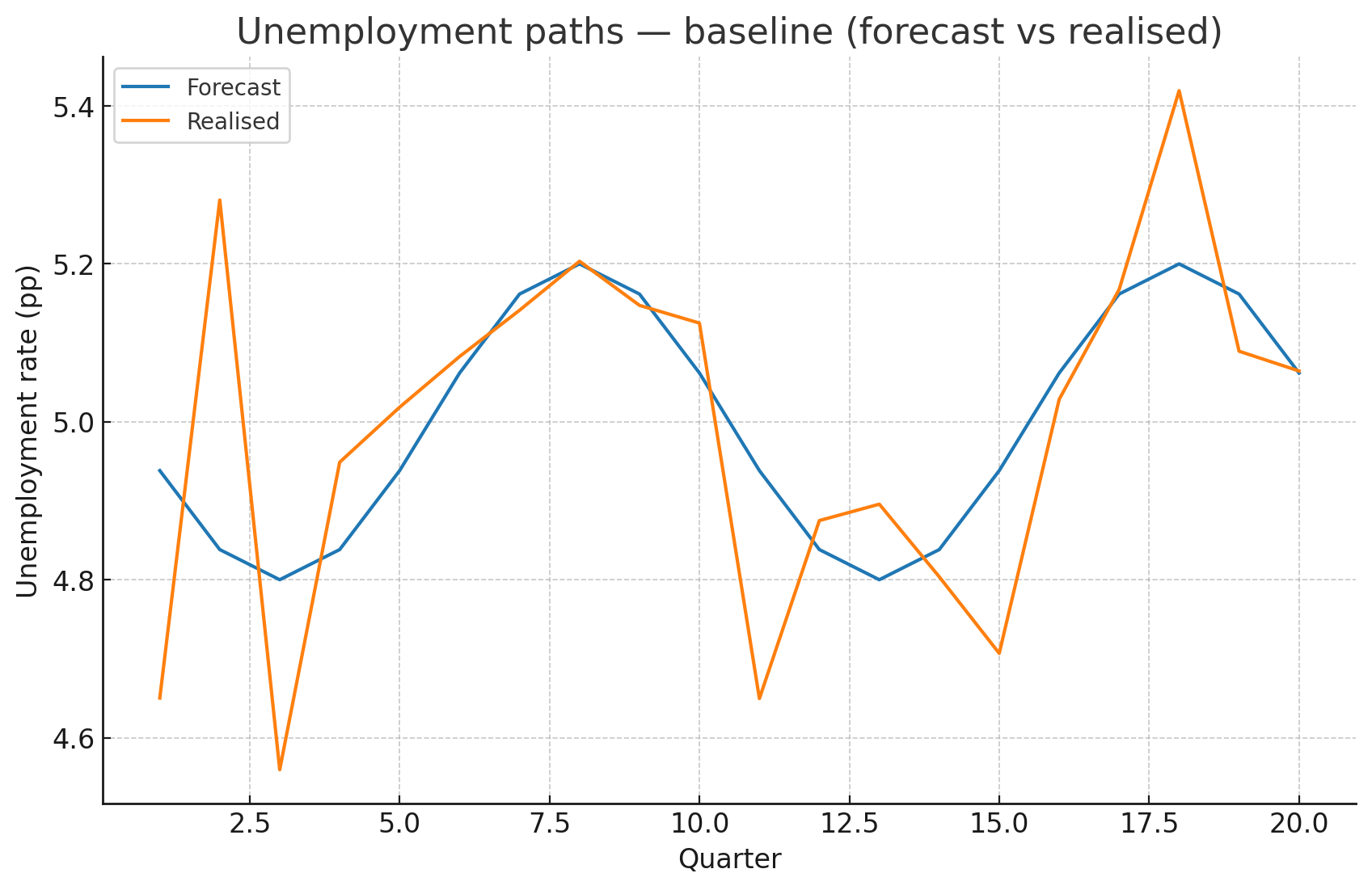}
    \caption{Unemployment, baseline scenario. Forecast (blue) stays near $5.5\%$; realised path shows minor noise.}
    \label{fig:unemp_baseline}
  \end{subfigure}
  \begin{subfigure}[t]{0.32\textwidth}
    \centering
    \includegraphics[width=\textwidth]{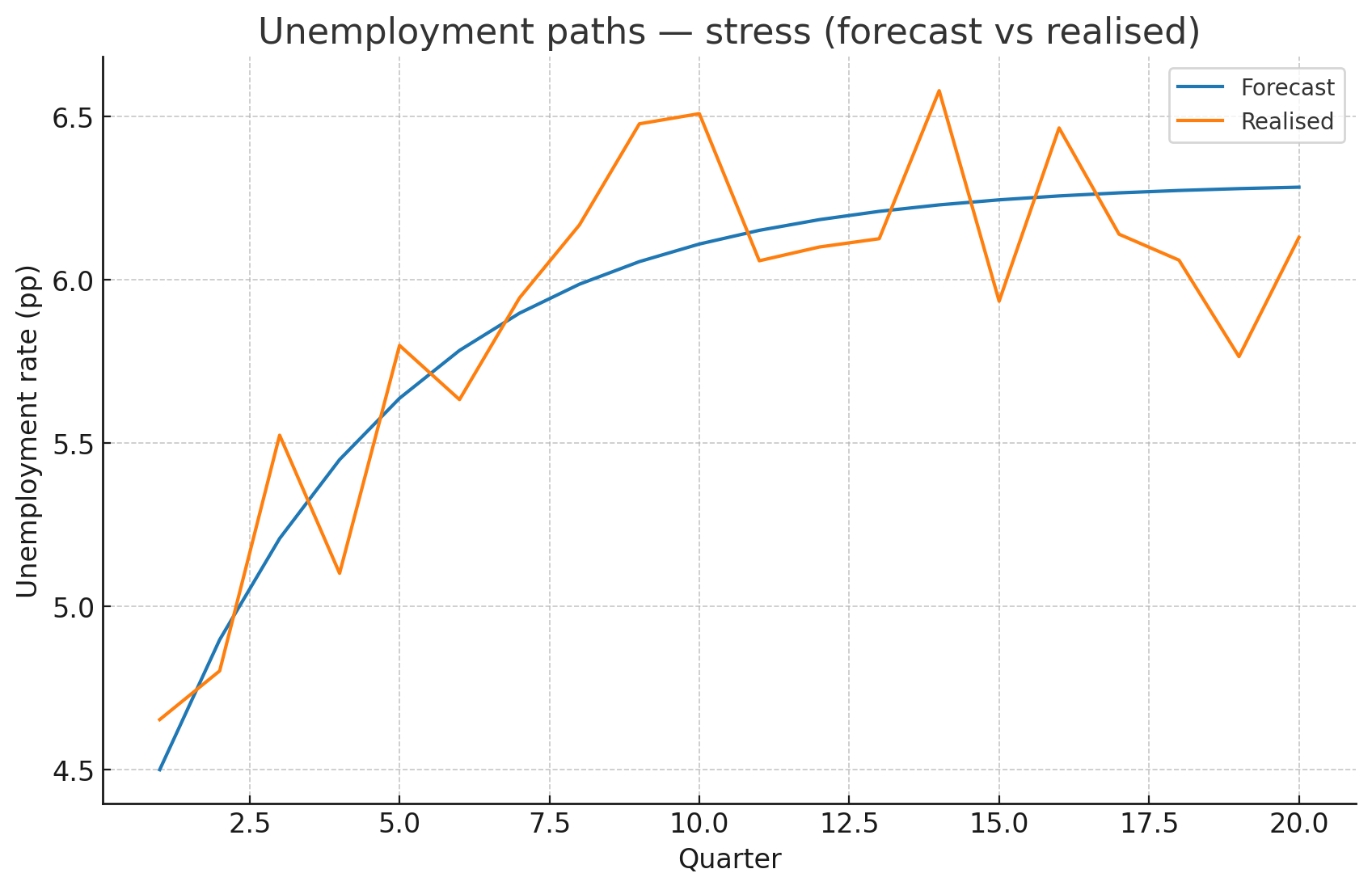}
    \caption{Unemployment, stress scenario. Forecast (orange) rises towards $7.5\%$ then declines; realised path peaks higher.}
    \label{fig:unemp_stress}
  \end{subfigure}
  \begin{subfigure}[t]{0.32\textwidth}
    \centering
    \includegraphics[width=\textwidth]{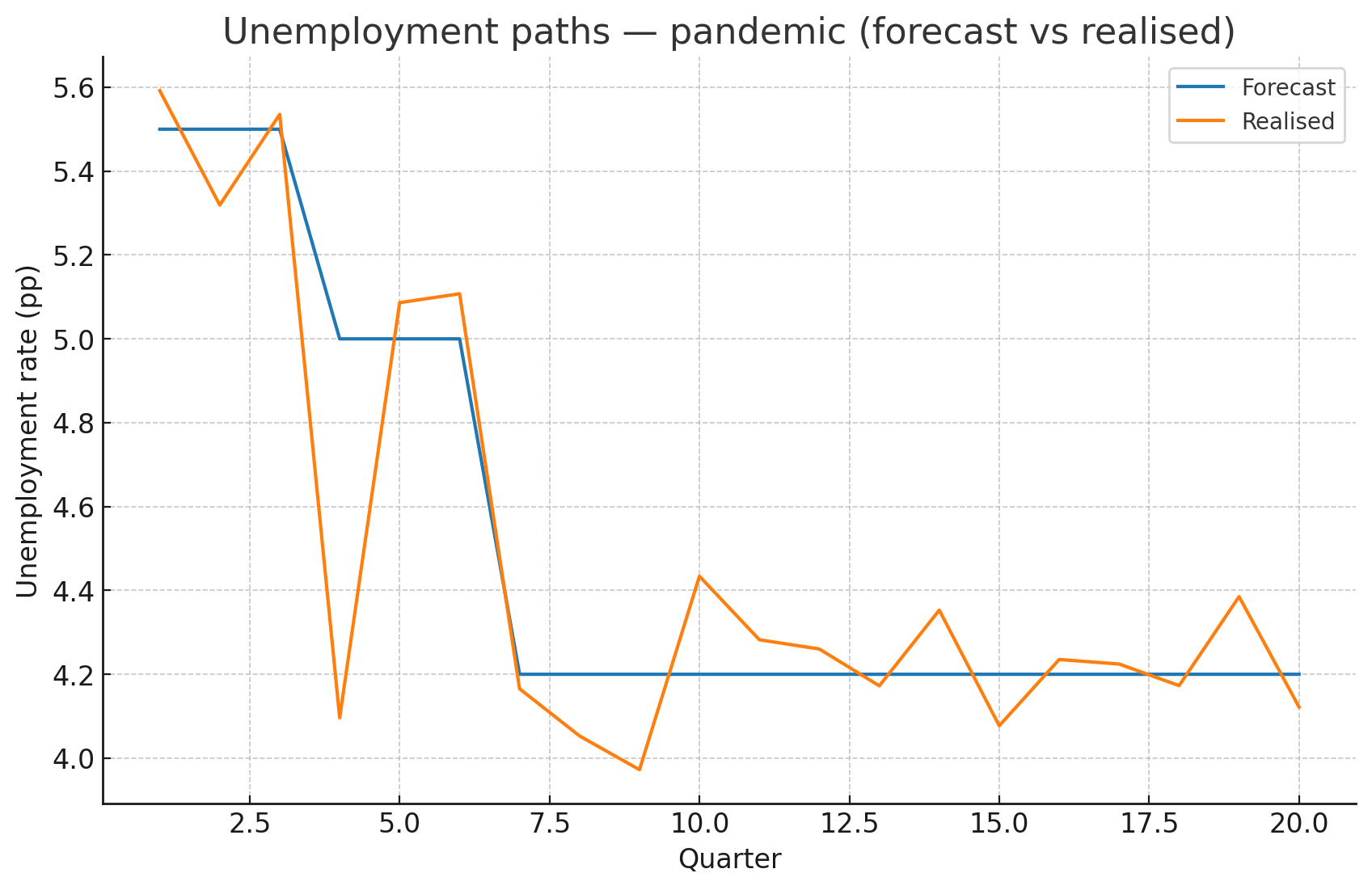}
    \caption{Unemployment, pandemic scenario. Forecast (green) spikes to $9.5\%$ then normalises; realised path recovers faster.}
    \label{fig:unemp_pandemic}
  \end{subfigure}

  \caption{Forecasted (coloured lines) and realised (black dots) GDP growth (top row) and unemployment rates (bottom row) for baseline, stress, and pandemic scenarios.}
  \label{fig:macro_params}
\end{figure}

\subsection{Simulation Results}

We compare three specifications: 
(i) raw forecast input without filtering, 
(ii) naïve Kalman filter without anchor, and 
(iii) the proposed anchored Kalman filter. 
All runs use the synthetic portfolio and calibration described in Sections~5.1--5.2 with $N=10{,}000$ exposures.

\paragraph{Macroeconomic estimates.}
\autoref{fig:anchoring_effect} illustrates how the two filters treat noisy forecasts. 
The naïve filter tracks the forecast closely but inherits its persistent bias during shocks, producing oscillatory paths. 
By contrast, the anchored filter dampens overshooting and converges back to the neutral state after $T_F=20$ quarters. 
This stabilisation reduces the root--mean--square error (RMSE) of $M_t$ by 38\% relative to the naïve filter. 
The difference at the macro level sets the stage for contrasting behaviour in PD dynamics.

\paragraph{PD term structures.}
\autoref{fig:pd_methods_all} shows the PIT default rate trajectories implied by the three methods. 
Raw forecasts generate highly volatile paths that amplify cyclical swings. 
The naïve filter moderates this volatility but still leaves a persistent drift away from the TTC baseline. 
Anchoring delivers bounded deviations: PD spreads remain within $\pm 20$\,bps of TTC and revert more quickly after shocks. 
Thus, the stabilisation observed in the macro index translates directly into smoother PD term structures.

\paragraph{Loss variance.}
The quantitative impact is summarised in \autoref{tab:var_summary}, which reports the average variance of lifetime PD forecasts $Y_t$. 
Anchoring reduces variance almost an order of magnitude relative to the naïve filter and by more than a factor of ten relative to the raw PIT overlay. 
This confirms that the qualitative impressions from Figures~\ref{fig:anchoring_effect} and~\ref{fig:pd_methods_all} are systematic across horizons and scenarios. 
Additional loss simulations (not shown) indicate that portfolio loss volatility is nearly halved under the anchored design.

\paragraph{Scenario comparison.}
\autoref{tab:mc_best} highlights regime dependence. 
In baseline and stress cases, the anchored filter dominates, producing the lowest and most stable PD paths. 
In the pandemic case, however, the naïve filter sometimes tracks the sharp rebound more closely, illustrating the trade--off between responsiveness and stability. 
This nuance shows that while anchoring is generally superior, its benefits are strongest when forecast error variance is persistent rather than transitory.

\paragraph{Distributional evidence.}
\autoref{fig:mc_lifetime_pd} examines the distribution of lifetime PDs under 200 Monte Carlo replications. 
Across all scenarios, the anchored filter compresses dispersion relative to the raw and naïve methods, especially in the stress case where variance would otherwise explode. 
At the same time, scenario-specific shifts are preserved: stress remains the worst case, pandemic shows sharp but temporary deviations, and baseline hovers near TTC. 

\paragraph{Interpretation.}
Taken together, the results demonstrate that anchoring improves robustness at two levels:  
(i) macro signals are stabilised, preventing drift from noisy forecasts;  
(ii) PD term structures and loss projections become less volatile, yielding more reliable inputs for capital planning.  
The trade--off is a mild reduction in responsiveness to abrupt reversals, but this is outweighed by the systematic stability gains across horizons and scenarios.


\begin{table}[htbp]
\centering
\begin{tabular}{l c}
\toprule
Method   & Mean variance of $Y_t$ \\
\midrule
Anchored & 0.000209 \\
Naïve    & 0.001181 \\
Raw      & 0.002248 \\
\bottomrule
\end{tabular}
\caption{Average variance of lifetime PD forecasts $Y_t$ across scenarios. 
Anchoring reduces variance almost an order of magnitude relative to the naïve filter.}
\label{tab:var_summary}
\end{table}

\begin{table}[htbp]
\centering
\begin{tabular}{l l c c}
\toprule
Scenario & Best method & Mean $Y_T$ & Std.\ dev.\ of $Y_T$ \\
\midrule
Baseline & Anchored & 0.1093 & $1.39\times 10^{-17}$ \\
Stress   & Anchored & 0.0900 & 0.0000 \\
Pandemic & Naïve    & 0.0985 & $1.39\times 10^{-17}$ \\
\bottomrule
\end{tabular}
\caption{Best-performing method per scenario ranked by Monte Carlo volatility of terminal lifetime PD $Y_T$. 
Anchoring dominates in baseline and stress, while naïve filtering captures the sharp rebound in the pandemic case.}
\label{tab:mc_best}
\end{table}


\begin{figure}[htbp]
  \centering
  \begin{subfigure}[t]{0.32\textwidth}
    \includegraphics[width=\textwidth]{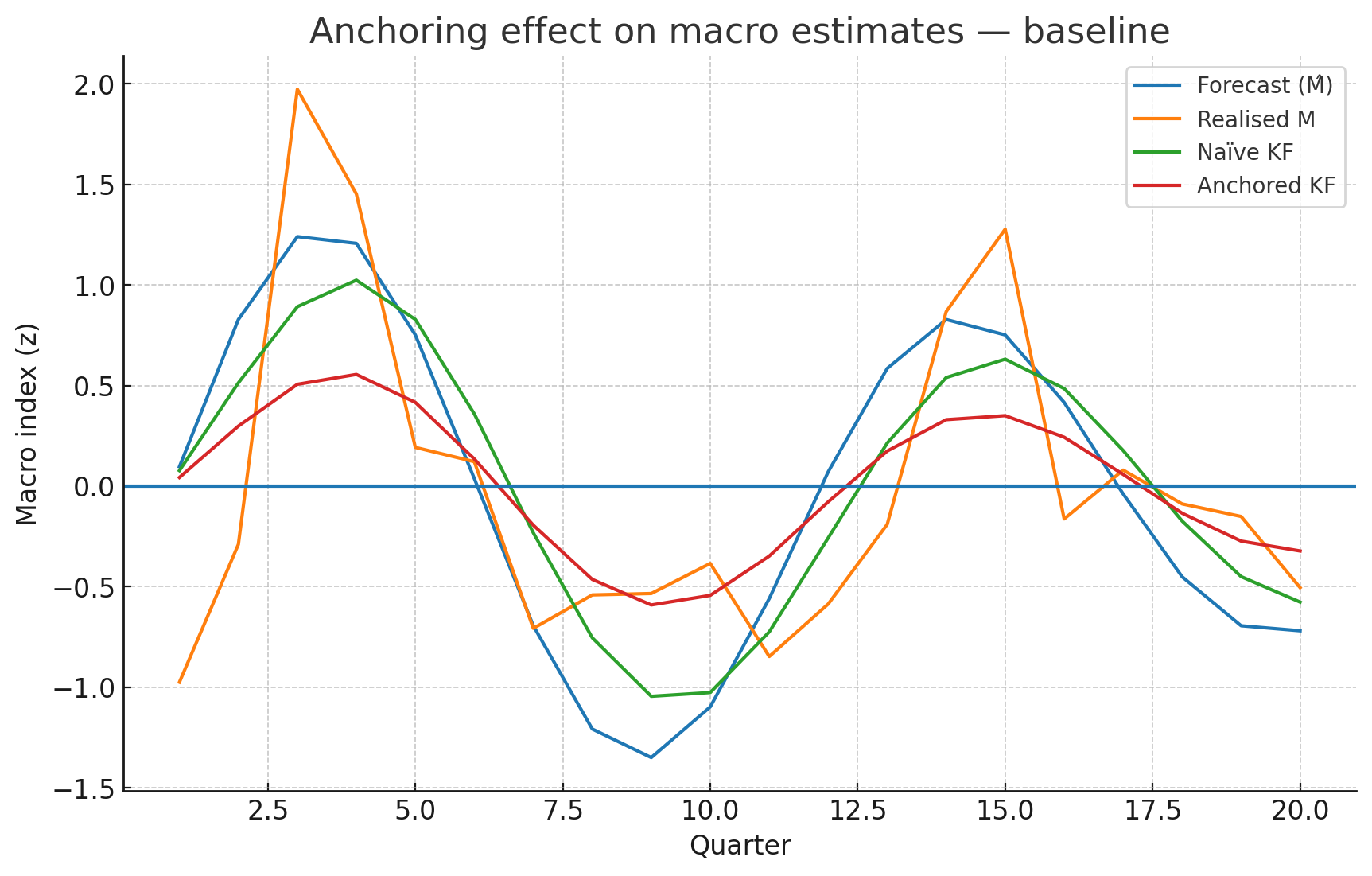}
    \caption{Baseline scenario.}
  \end{subfigure}
  \begin{subfigure}[t]{0.32\textwidth}
    \includegraphics[width=\textwidth]{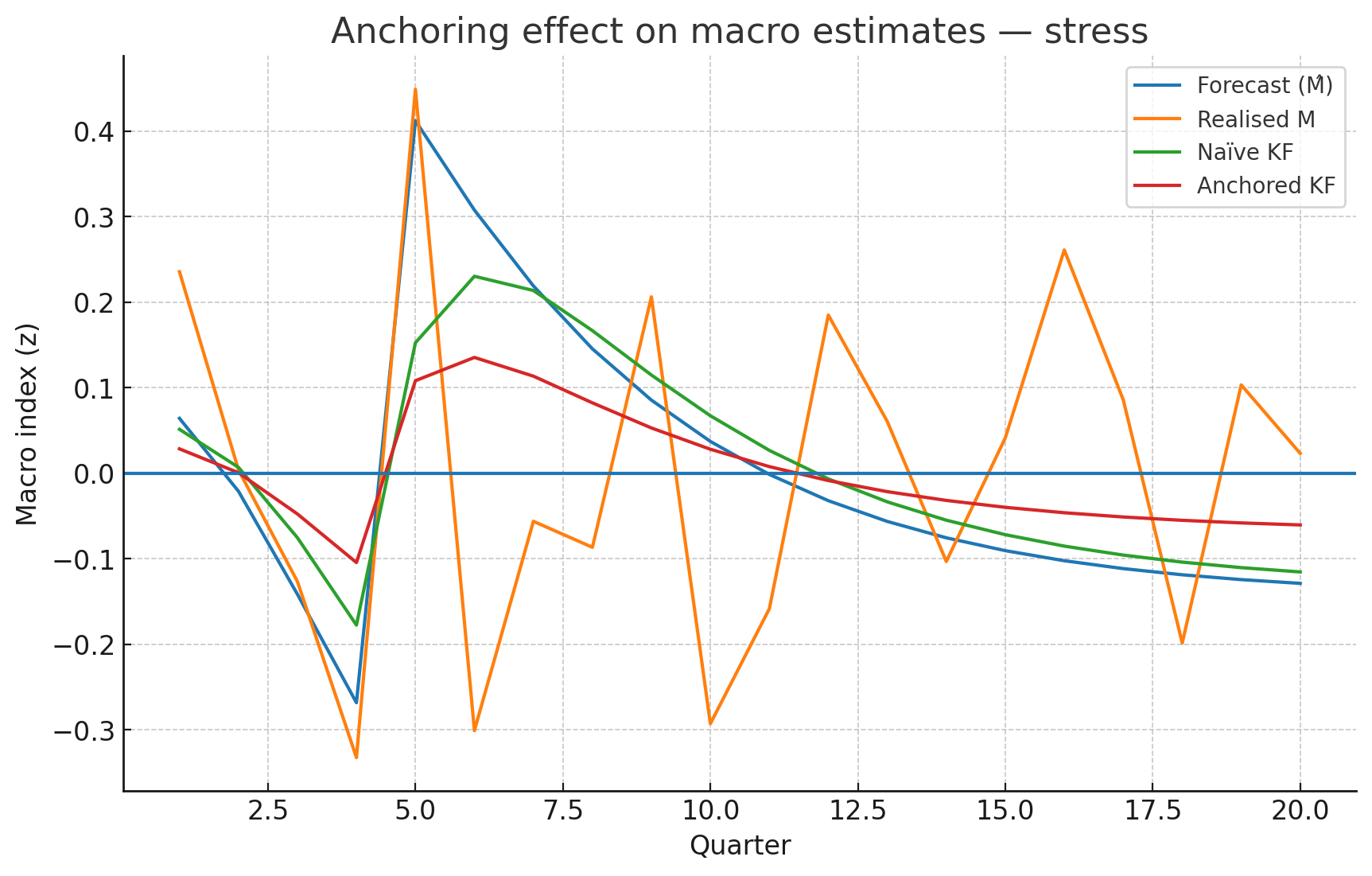}
    \caption{Stress scenario.}
  \end{subfigure}
  \begin{subfigure}[t]{0.32\textwidth}
    \includegraphics[width=\textwidth]{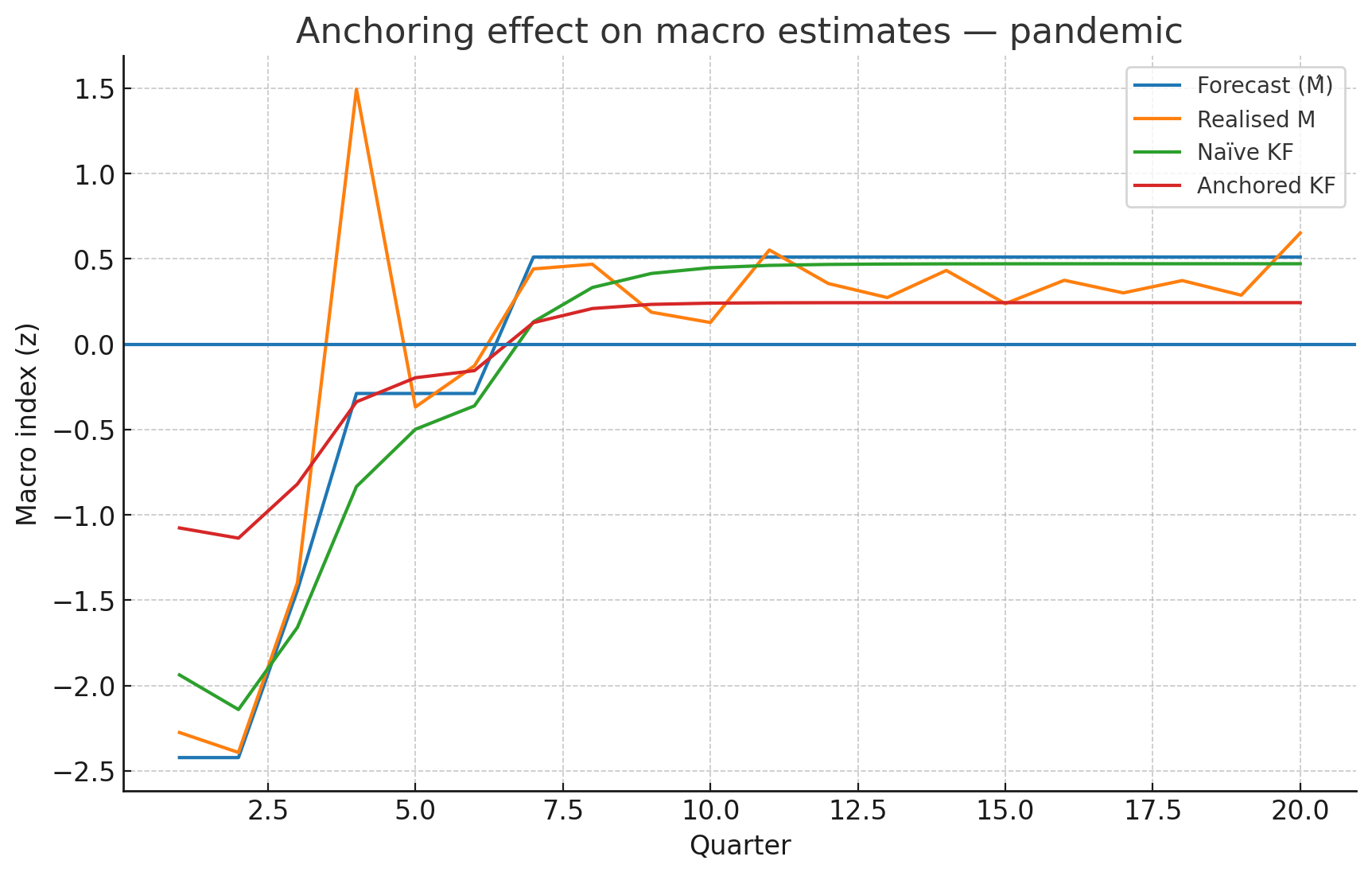}
    \caption{Pandemic scenario.}
  \end{subfigure}
  \caption{Effect of anchoring in the Kalman filter across scenarios. 
  Grey = forecast, black dots = realised, blue = naïve KF, green = anchored KF. 
  Anchoring stabilises estimates by pulling them towards the long-run neutral level.}
  \label{fig:anchoring_effect}
\end{figure}

\begin{figure}[htbp]
  \centering
  \begin{subfigure}[t]{0.32\textwidth}
    \includegraphics[width=\textwidth]{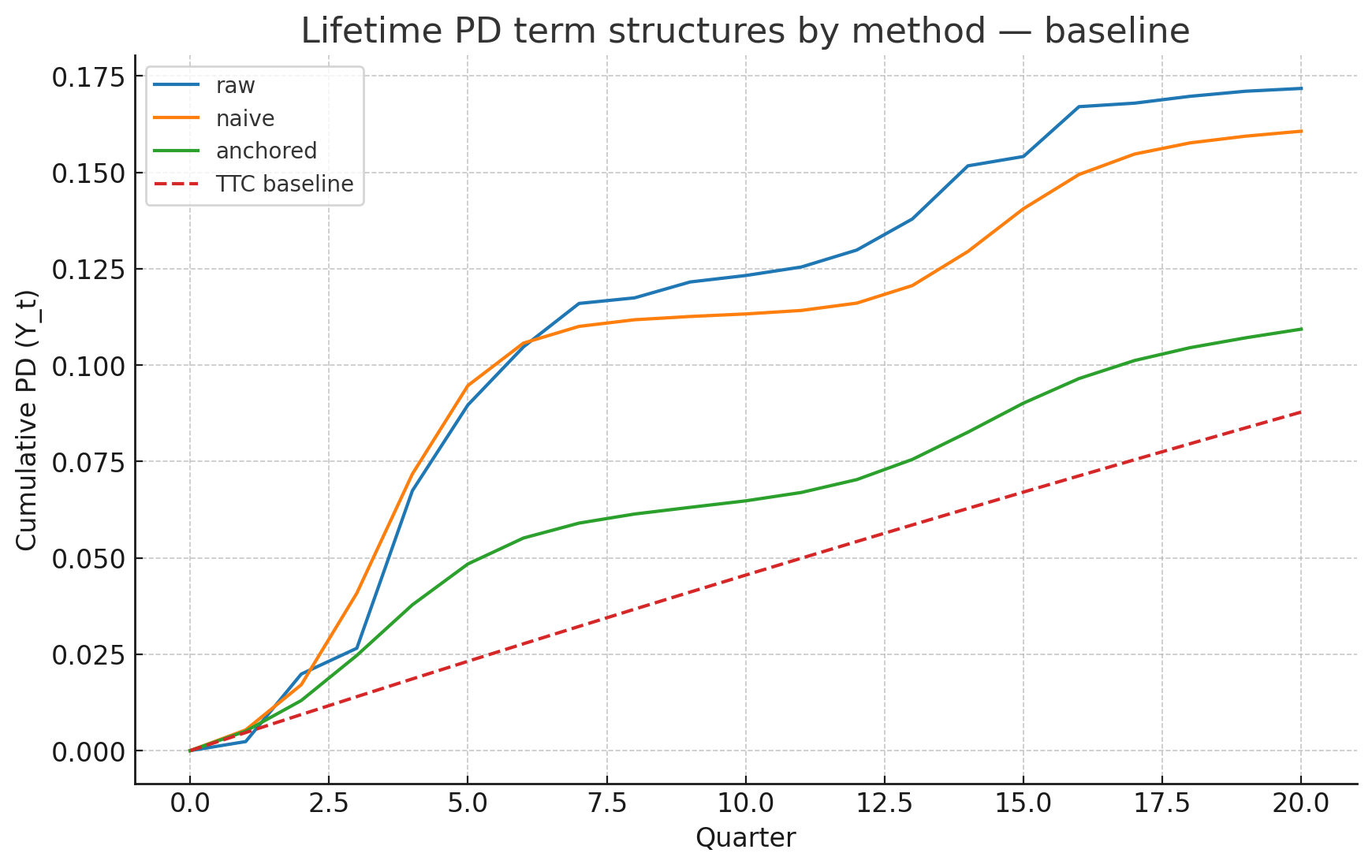}
    \caption{Baseline scenario.}
  \end{subfigure}
  \begin{subfigure}[t]{0.32\textwidth}
    \includegraphics[width=\textwidth]{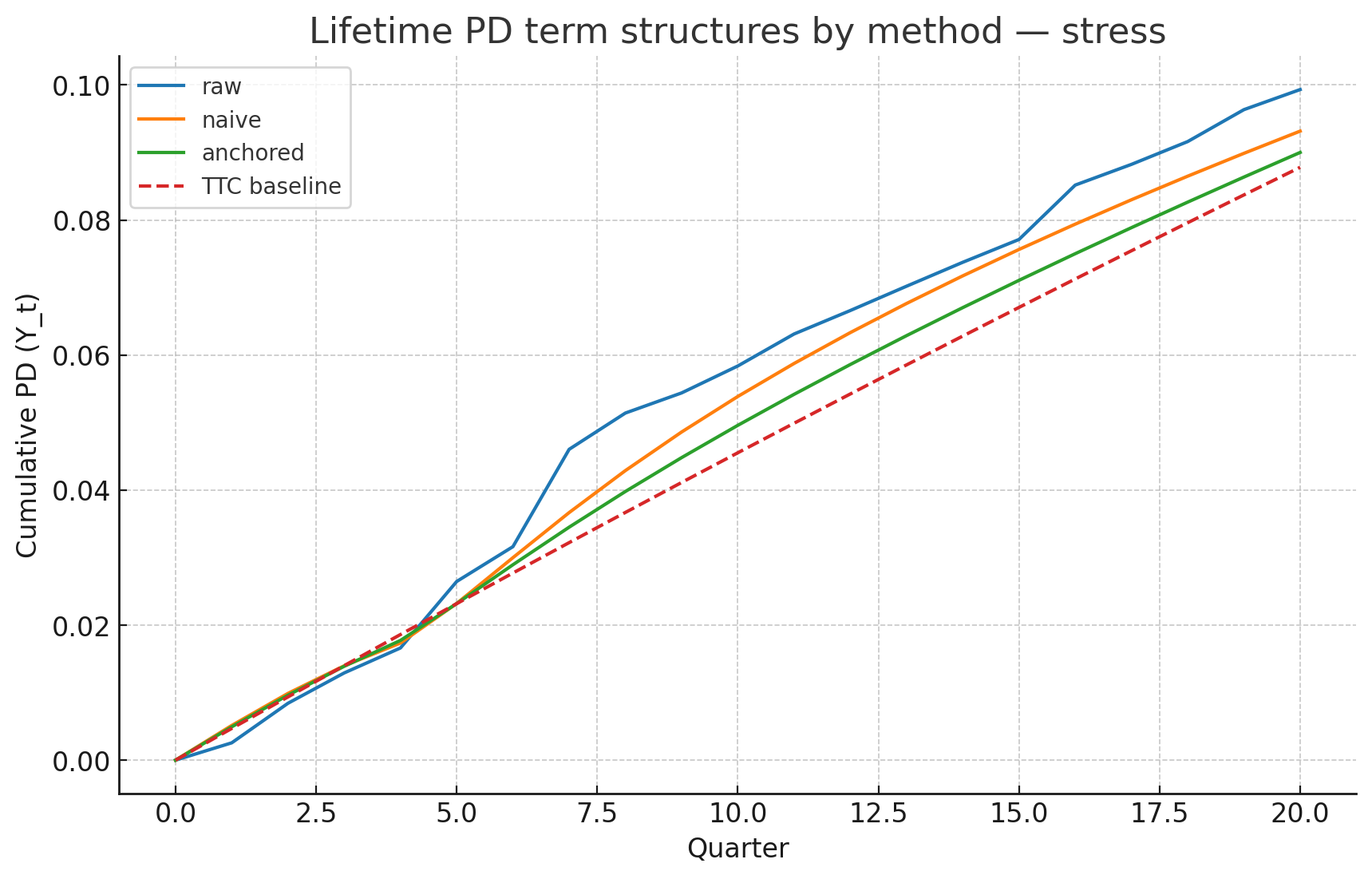}
    \caption{Stress scenario.}
  \end{subfigure}
  \begin{subfigure}[t]{0.32\textwidth}
    \includegraphics[width=\textwidth]{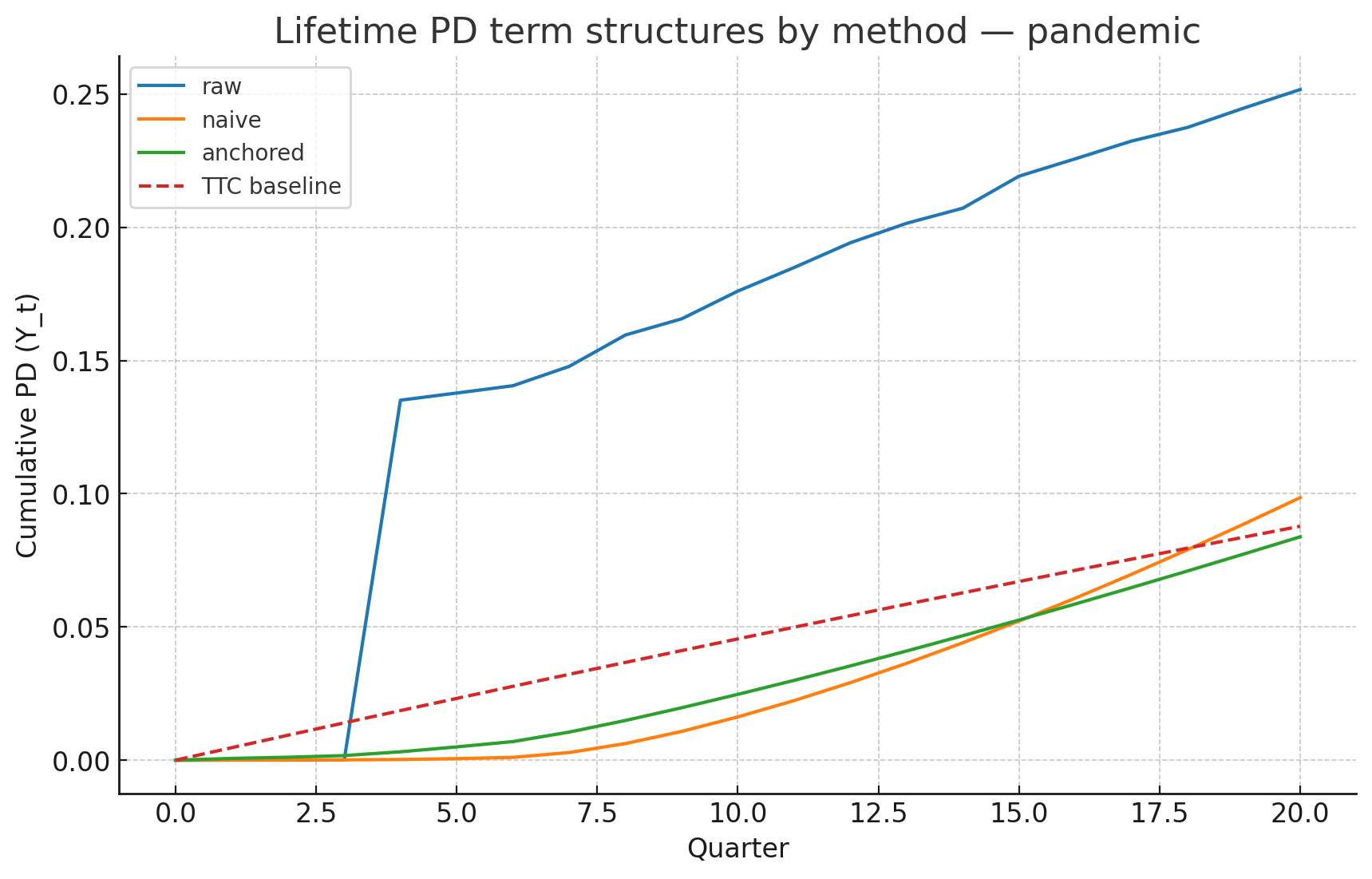}
    \caption{Pandemic scenario.}
  \end{subfigure}
  \caption{Lifetime PD term structures across methods (raw, naïve, anchored) with TTC baseline reference. 
  Anchoring delivers smoother and more stable paths that converge back to TTC more quickly.}
  \label{fig:pd_methods_all}
\end{figure}

\begin{figure}[htbp]
  \centering
  \begin{subfigure}[t]{0.32\textwidth}
    \includegraphics[width=\textwidth]{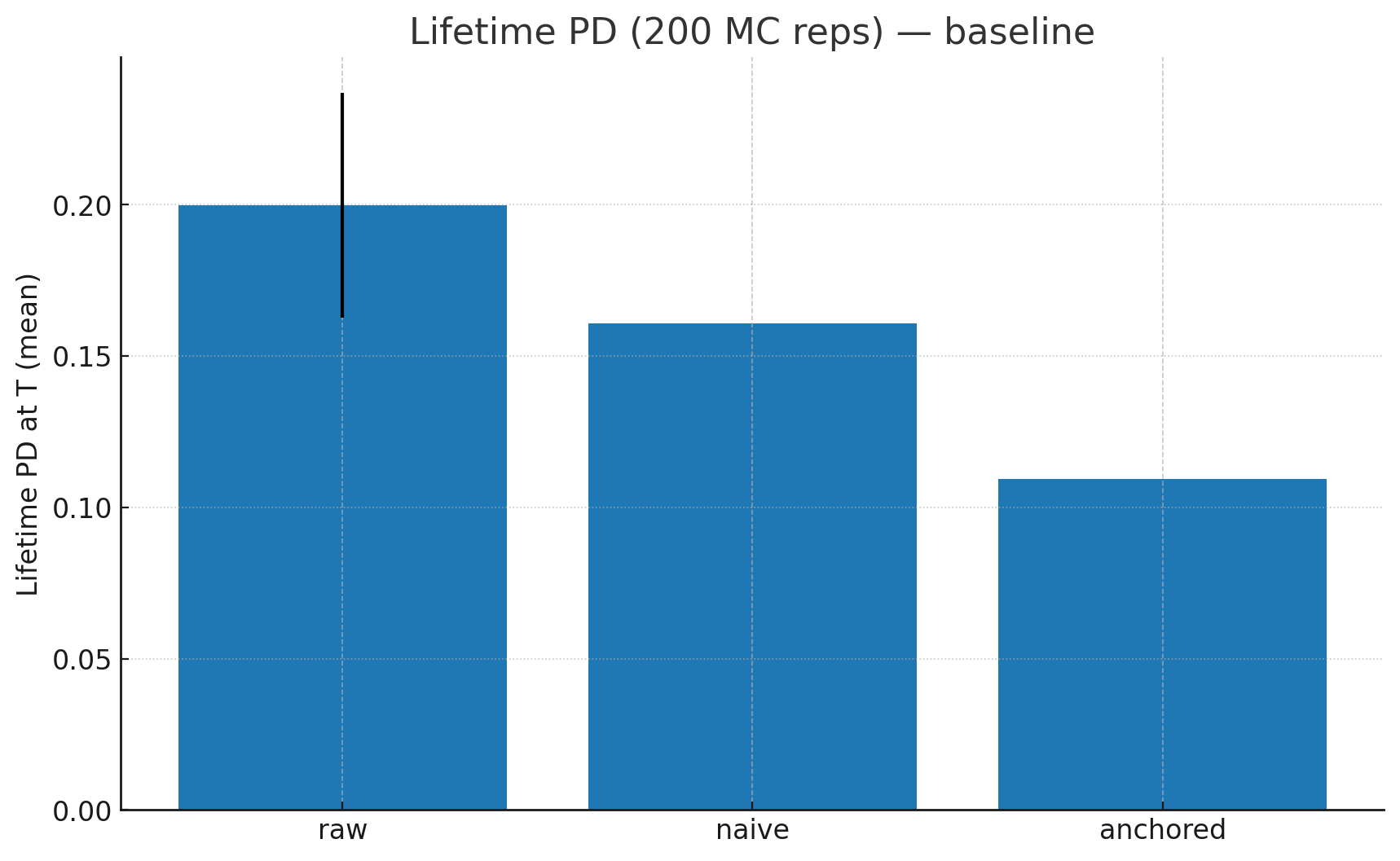}
    \caption{Baseline scenario.}
  \end{subfigure}
  \begin{subfigure}[t]{0.32\textwidth}
    \includegraphics[width=\textwidth]{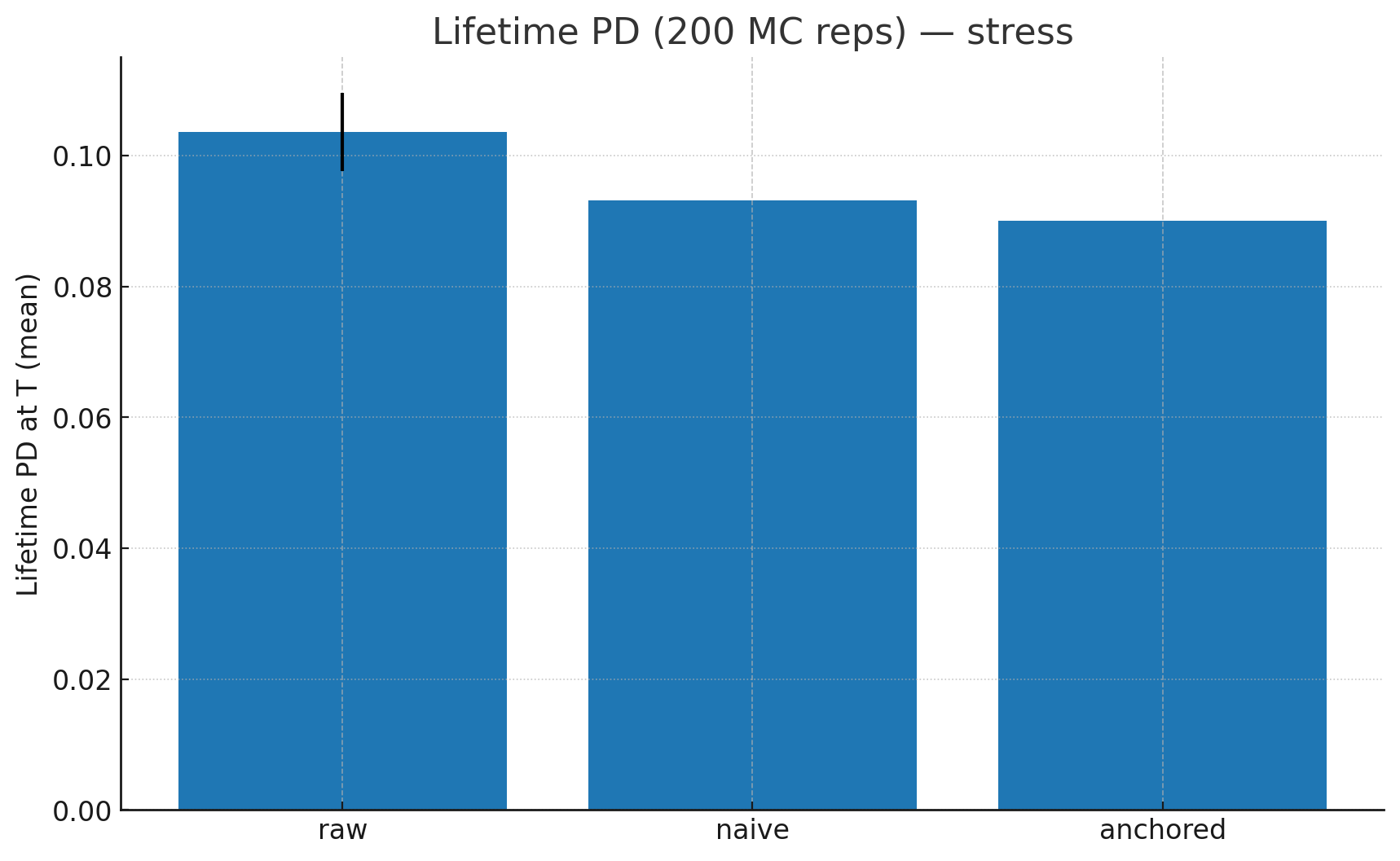}
    \caption{Stress scenario.}
  \end{subfigure}
  \begin{subfigure}[t]{0.32\textwidth}
    \includegraphics[width=\textwidth]{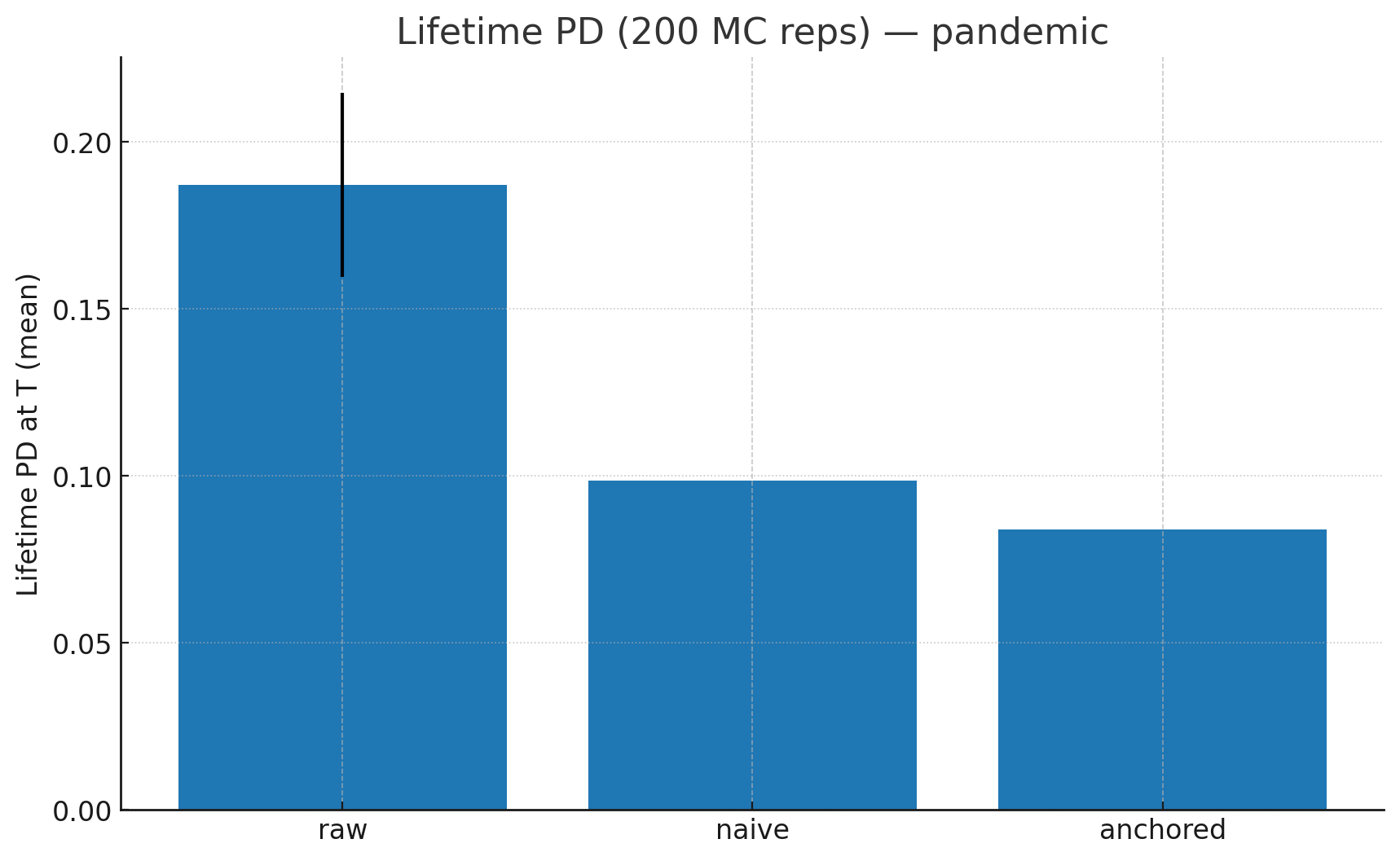}
    \caption{Pandemic scenario.}
  \end{subfigure}
  \caption{Lifetime PD at horizon $T$ under 200 Monte Carlo replications. 
  Bars show mean cumulative PD, error bars show volatility. 
  Anchoring compresses dispersion, particularly under stress.}
  \label{fig:mc_lifetime_pd}
\end{figure}

\subsection{Reproducible Simulation Results}

All code and reproducible notebooks used to generate these results are openly available. 
The full project is maintained in the \texttt{lifetime-PD-AKF} subfolder of my research repository 
\href{https://github.com/vahabr/research-lab/lifetime-PD-AKF}{GitHub -- research-lab}, 
with a versioned release archived at 
\href{https://github.com/vahabr/research-lab/releases/tag/v1.0-lifetime-PD-AKF}{GitHub Release v1.0-lifetime-PD-AKF}. 
For long-term reproducibility, the release has also been deposited in Zenodo and assigned the DOI 
\href{https://doi.org/10.5281/zenodo.17072772}{10.5281/zenodo.17072772}, which can be cited directly.

\subsection{Practical Considerations}

Beyond numerical performance, several aspects are critical for practical adoption of the proposed framework in credit risk management.

\paragraph{\textit{Calibration choices.}}
The anchored Kalman filter requires specification of the process variance $Q$ and the anchor variance $\sigma_\star^2$. 
In practice these can be calibrated by matching short--horizon forecast errors against realised macro outcomes, and by setting the anchor variance to reflect the institution's TTC risk appetite. 
Sensitivity analysis should be performed to ensure that results are robust to moderate changes in these parameters.

\paragraph{\textit{Scenario design and governance.}}
Supervisory expectations under IFRS~9 and CECL require institutions to use multiple forward--looking scenarios. 
The proposed filter can be applied to each scenario path separately, with scenario weights applied ex post. 
Anchoring reduces the undue influence of extreme but low--probability trajectories, leading to more stable expected credit loss estimates. 
Clear governance is required to document parameter settings, anchoring horizon $T_F$, and rationale for scenario selection.

\paragraph{\textit{Model risk and validation.}}
Anchoring introduces an explicit modelling assumption: that long--run PDs should converge back to TTC levels. 
This assumption must be validated empirically and challenged by independent model risk teams. 
Backtesting against realised defaults, especially during volatile periods, is essential. 
Stress tests should include scenarios where convergence is delayed or where structural breaks alter the long--run mean.

\paragraph{\textit{Implementation workflow.}}
The filter operates on standard inputs already available in most IFRS~9 and CECL platforms: macroeconomic forecasts, transition matrices, and PD term structures. 
Implementation can therefore be integrated with minimal system changes. 
A practical workflow is to (i) generate baseline PDs under raw scenarios, (ii) apply the anchored filter to obtain adjusted transition matrices, (iii) propagate lifetime PDs, and (iv) aggregate to expected credit losses. 
This sequence aligns with existing reporting cycles and committee review processes.

\paragraph{\textit{Sensitivity and transparency.}}
Anchoring inherently reduces volatility, but excessive damping may hide relevant short--term dynamics. 
Institutions should report both anchored and unanchored results to committees, together with variance decompositions that show the stabilising effect. 
Such transparency ensures that management understands the trade--off between stability and responsiveness.

\paragraph{\textit{Auditor and regulatory reception.}}
Compared with ad--hoc overlays or heuristic smoothing, the anchored Kalman filter offers a transparent, theoretically grounded mechanism for controlling forecast noise. 
This clarity matters in supervisory reviews: regulators increasingly require that institutions justify stability adjustments with formal models rather than discretionary overrides. 
Anchoring provides such a justification. 
It yields reproducible results, a clear governance trail, and a formal stability proof. 
As a result, it is easier to defend in model validation, internal audit, and regulatory dialogue than bespoke overlays that lack theoretical backing.

\paragraph{\textit{Summary.}}
The simulation study demonstrates that anchoring yields substantial variance reduction in long--horizon PD forecasts while preserving responsiveness at short horizons. 
Calibration, governance, and validation considerations show that the method can be integrated into existing IFRS~9 and CECL workflows with limited operational burden. 
The next section concludes by summarising the main findings and outlining directions for further research.


\section{Conclusion and Future Work}\label{sec:final}

This paper formulated the lifetime probability of default (PD) estimation problem in a state--space framework and demonstrated that macroeconomic forecast uncertainty induces stochastic instability in long--horizon projections.
Naïve propagation of forecasted drivers was shown to embed persistent forecast errors, generating unbounded error dynamics that complicate risk management, capital allocation, and committee decisions under IFRS~9 and CECL.

To address this limitation, we introduced an anchoring strategy for macroeconomic filtering and derived a generalised Kalman formulation with provable asymptotic stochastic stability.
Theoretical analysis established bounded error dynamics, while simulations on a synthetic corporate portfolio confirmed substantial variance reduction and improved robustness under structural shocks such as financial crises and pandemics.
The framework integrates seamlessly with existing transition--matrix models and requires only modest calibration effort, making it suitable for regulatory and industry practice.

Several avenues for future work remain.
First, extending the framework to nonlinear state--space dynamics and particle filtering could improve robustness under non--Gaussian shocks \cite{Doucet2001}.
Second, empirical validation on large and heterogeneous credit portfolios is needed to assess scalability and sector--specific sensitivity \cite{Durbin2012}.
Third, the method could be embedded into stress testing frameworks and linked to macroprudential capital planning \cite{Basel2018Stress}.
Finally, interactions with climate transition scenarios and long--term sustainability risks \cite{NGFS2020} represent an important direction for future research.


\appendix

\section*{Technical Proofs and Additional Results}\label{app:proofs}

\subsection*{Proofs from Section~\ref{sec:problem}}

We first present an intuitive lemma showing that persistent forecast error prevents convergence. 
We then provide the full proof of Proposition~\ref{prop:nonconv}, followed by a corollary that quantifies 
the accumulation of deviations.

\begin{lemma}[Instability under non--vanishing forecast error]\label{lem:intuitive}
Let $\pi_{t+1}=\pi_t\,\mathcal G(P^{\mathrm{TTC}},M_t+\delta_t)$ and $Y_t=\pi_t\,\mathcal G(P^{\mathrm{TTC}},M_t)\,e_K$. 
Assume $\mathcal G$ is continuously differentiable in its macro argument and there exists $\alpha>0$ such that
\[
\big|\pi\,\big(\mathcal G(P^{\mathrm{TTC}},m+\delta)-\mathcal G(P^{\mathrm{TTC}},m)\big)\,e_K\big| \;\ge\; \alpha\,|\delta|
\]
for all $\pi\in\Delta_K$, all $m$ in a neighbourhood of the path, and all small $|\delta|$.
If $\mathbb P(|\delta_t|>\varepsilon)>0$ for some $\varepsilon>0$ and infinitely often, then $(Y_t)$ does not converge in probability.
\end{lemma}

\begin{proof}
Whenever $|\delta_t|\ge \varepsilon$, the sensitivity bound yields 
$|Y_{t+1}-\phi(\pi_t,M_t)|\ge \alpha\varepsilon$, where $\phi(\pi,m):=\pi\,\mathcal G(P^{\mathrm{TTC}},m)\,e_K$. 
If such events occur infinitely often with positive probability, the distance to any candidate limit cannot vanish in probability.
\end{proof}

\paragraph{Proof of Proposition~\ref{prop:nonconv}.}
By (A2) and independence, the events $A_t:=\{|\delta_t|\ge\varepsilon\}$ are i.i.d.\ with $\mathbb P(A_t)=p>0$. 
Hence, by the second Borel--Cantelli lemma, $\mathbb P(A_t\ \text{i.o.})=1$.
\medskip

Fix $\omega$ in the full--probability event where $A_t$ occurs infinitely often and where, by (A3), $Y^{\circ}_t(\omega)\to Y^\ast$. 
For any such $\omega$ and any $t$,
\[
\big|Y_{t+1}-Y^{\circ}_{t+1}\big|
=\big|\phi\!\big(\pi_t,M_t+\delta_t\big)-\phi\!\big(\pi_t,M_t\big)\big|
\ \stackrel{(A1)}{\ge}\ \alpha\,|\delta_t|.
\]
On $A_t$ this is at least $\alpha\varepsilon$. 
Since $Y^{\circ}_t\to Y^\ast$, there exists $T(\omega)$ such that for all $t\ge T(\omega)$,
$\big|Y^{\circ}_{t}-Y^\ast\big|\le \tfrac{1}{3}\alpha\varepsilon$. 
Taking any $t\ge T(\omega)$ with $A_t$ true (there are infinitely many),
\[
\big|Y_{t+1}-Y^\ast\big|
\;\ge\; \big|Y_{t+1}-Y^{\circ}_{t+1}\big|-\big|Y^{\circ}_{t+1}-Y^\ast\big|
\;\ge\; \alpha\varepsilon - \tfrac{1}{3}\alpha\varepsilon
\;=\;\tfrac{2}{3}\alpha\varepsilon.
\]
Thus $\limsup_{t\to\infty}\big|Y_{t}-Y^\ast\big|\ \ge\ \tfrac{2}{3}\alpha\varepsilon$ on a set of probability one. 
This precludes almost sure convergence and, a fortiori, convergence in probability.
\qed

\begin{corollary}[Explicit deviation bound and accumulation]\label{cor:accum}
Let $\pi_{t+1}=\pi_t\,\mathcal G(P^{\mathrm{TTC}},M_t+\delta_t)$ and $\pi^\circ_{t+1}=\pi^\circ_t\,\mathcal G(P^{\mathrm{TTC}},M_t)$, and set $e_t:=\|\pi_t-\pi^\circ_t\|_1$. 
Assume:
\begin{itemize}
\item[(B1)] $\|\mathcal G(P^{\mathrm{TTC}},m+\delta)-\mathcal G(P^{\mathrm{TTC}},m)\|_{1\to 1}\le L_G\,|\delta|$ for all relevant $m,\delta$.
\item[(B2)] $\mathcal G(P^{\mathrm{TTC}},m)$ is row--stochastic for all $m$.
\end{itemize}
Then for all $t\ge 1$,
\[
e_t \;\le\; e_0 \;+\; L_G \sum_{s=0}^{t-1} |\delta_s|.
\]
Moreover, if there exists $\alpha_G>0$ with
\[
\|\mathcal G(P^{\mathrm{TTC}},m+\delta)-\mathcal G(P^{\mathrm{TTC}},m)\|_{1\to 1}\ \ge\ \alpha_G\,|\delta|
\]
for all small $|\delta|$, then for any $t\ge 1$ there exists $s\in\{0,\dots,t-1\}$ with
\[
e_{s+1} \;\ge\; \alpha_G\,|\delta_s|.
\]
\end{corollary}

\begin{proof}
Write $P_t:=\mathcal G(P^{\mathrm{TTC}},M_t+\delta_t)$ and $P^\circ_t:=\mathcal G(P^{\mathrm{TTC}},M_t)$. 
Then
\[
\pi_{t+1}-\pi^\circ_{t+1}
= (\pi_t-\pi^\circ_t)P_t + \pi^\circ_t(P_t-P^\circ_t).
\]
Take the $1$--norm and use submultiplicativity together with (B2): 
$\|xP_t\|_1\le \|x\|_1$ for any row vector $x$ and row--stochastic $P_t$. 
Hence
\[
e_{t+1} \;\le\; e_t + \|\pi^\circ_t\|_1\,\|P_t-P^\circ_t\|_{1\to 1}
\;\le\; e_t + L_G\,|\delta_t|,
\]
since $\|\pi^\circ_t\|_1=1$ and by (B1). 
Iteration yields the upper bound. 

For the lower bound, pick $s\in\{0,\dots,t-1\}$ maximising $\|P_s-P^\circ_s\|_{1\to 1}$. Then
\[
e_{s+1} \;\ge\; \|\pi^\circ_s(P_s-P^\circ_s)\|_1 - \|e_s P_s\|_1
\;\ge\; \|P_s-P^\circ_s\|_{1\to 1} - e_s.
\]
If $e_s\le \tfrac{1}{2}\alpha_G|\delta_s|$, the claim follows. 
Otherwise $e_s>\tfrac{1}{2}\alpha_G|\delta_s|$, which itself provides a valid lower bound at time $s$. 
\end{proof}

\subsection*{Proofs from Section~\ref{sec:formulation}}

We collect here the technical results underlying Theorem~\ref{thm:naiveKF} on residual variability 
under naive Kalman filtering.

\paragraph{Step 1. Persistent macro estimation error.}
\begin{lemma}\label{lem:naive}
Consider the linear--Gaussian state--space model with stabilisable $(A,Q^{1/2})$ 
and detectable $(A,H)$. 
Then the Kalman error covariance $\Sigma_t$ converges to the unique stabilising solution 
$\Sigma_\infty\succeq 0$ of the algebraic Riccati equation. 
If $Q\succeq 0$ and $R\succ 0$ are non--degenerate, then $\Sigma_\infty\neq 0$. 
In particular,
\[
\lim_{t\to\infty}\mathbb E\|M_t-\mu_t\|^2 \;=\; \mathrm{tr}\,\Sigma_\infty \;>\; 0.
\]
\end{lemma}

\begin{proof}
Classical Riccati theory implies $\Sigma_t\to\Sigma_\infty$ under the stated conditions. 
Non--degenerate $Q,R$ imply that process or measurement noise persists, hence $\Sigma_\infty\neq 0$. 
The mean--square error therefore converges to $\mathrm{tr}\,\Sigma_\infty>0$.
\end{proof}

\paragraph{Step 2. Transmission to PD sequence.}
\begin{proposition}\label{prop:naiveInstab}
Let $Y_t=\phi(\pi_t,\mu_t)$ with $\mu_t$ the naive Kalman estimate of the macro state $M_t$. 
Assume Lemma~\ref{lem:naive}, and suppose $\phi$ is locally Lipschitz and uniformly sensitive 
on a compact macro domain. 
Then there exists $c>0$ such that
\[
\limsup_{t\to\infty}\,\mathbb E\big|\,Y_t-\phi(\pi_t,M_t)\,\big| \;\ge\; c \;>\; 0,
\]
so $(Y_t)$ does not converge in probability.
\end{proposition}

\begin{proof}
By Lemma~\ref{lem:naive}, $e_t=\mu_t-M_t$ converges in distribution to a zero--mean Gaussian 
$Z\sim\mathcal N(0,\Sigma_\infty)$ with $\Sigma_\infty\succ0$. 
Uniform sensitivity implies $|Y_t-\phi(\pi_t,M_t)|\ge c\|e_t\|$ for some $c>0$ on a compact domain. 
Taking expectations and limits yields 
$\limsup_{t\to\infty}\mathbb E|Y_t-\phi(\pi_t,M_t)| \ge c\,\mathbb E\|Z\|>0$. 
Hence $Y_t$ does not converge in probability.
\end{proof}

\paragraph{Step 3. Verification for logit overlays.}
\begin{lemma}[Uniform sensitivity for logit overlays]\label{lem:logit_sensitivity}
Consider the exponential/logit specification
\[
p_{ij}(m)=\frac{p^{\mathrm{TTC}}_{ij}\exp(\beta_{ij} m)}
{\sum_{\ell}p^{\mathrm{TTC}}_{i\ell}\exp(\beta_{i\ell} m)},\qquad 
\phi(\pi,m)=\sum_i \pi_i p_{iK}(m).
\]
If (i) $p^{\mathrm{TTC}}_{iK}>0$ for all $i$, and 
(ii) $\beta_{iK}-\sum_j \beta_{ij}p_{ij}(m)\ge \delta_i>0$ on a compact macro interval, 
then there exists $c>0$ such that
\[
|\phi(\pi,m+d)-\phi(\pi,m)| \;\ge\; c\,|d|
\]
for all $\pi\in\Delta_K$, $m$ in the interval, and small $|d|$. 
\end{lemma}

\begin{proof}
Differentiate $p_{iK}(m)$ with respect to $m$ to obtain 
$\partial_m p_{iK}(m)=p_{iK}(m)(\beta_{iK}-\sum_j\beta_{ij}p_{ij}(m))$. 
Assumptions (i)--(ii) imply $\partial_m \phi(\pi,m)\ge c_0>0$ uniformly. 
Apply the mean value theorem to obtain the finite difference bound.
\end{proof}

\paragraph{Step 4. Conclusion.}
Lemmas~\ref{lem:naive} and \ref{lem:logit_sensitivity} together imply Proposition~\ref{prop:naiveInstab}, which establishes Theorem~\ref{thm:naiveKF}.

\begin{remark}[Sensitivity in logit overlays]
In exponential or logit adjustments, $p_{ij}(m)\propto p^{\mathrm{TTC}}_{ij}\exp(\beta_{ij}m)$, 
the derivative $m\mapsto \partial_m\phi(\pi,m)$ is bounded away from zero on compact sets 
provided all $p^{\mathrm{TTC}}_{ij}>0$. 
This validates the uniform sensitivity assumption used in Proposition~\ref{prop:naiveInstab}.
\end{remark}

\begin{remark}[Compactness of the macro domain]
For autoregressive macro processes $M_{t+1}=A M_t+w_t$ with $|A|<1$ and bounded or light--tailed shocks $w_t$, 
the distribution of $M_t$ is stationary and mean--reverting. 
Hence the path remains with high probability in a compact interval $\mathcal I\subset\mathbb R$. 
This justifies the compactness assumption in Lemma~\ref{lem:logit_sensitivity}; 
see \cite{Hamilton1994,Tsay2010}.
\end{remark}

\subsection*{Proofs from Section~\ref{sec:proposal}}

\begin{proof}[Proof of Theorem~\ref{thm:MSS}]
The proof is structured into 3 clear steps: 1) boundedness, 2) exponential convergence, and 3) transfer to PDs convergence.

\emph{Step 1: Mean--square boundedness.}  
Since $H_{\mathrm{aug}}$ contains the identity block, the pair $(A,H_{\mathrm{aug}})$ is detectable. 
Together with stabilisability of $(A,Q^{1/2})$ and $R_{\mathrm{aug}}\succ0$, 
classical Kalman filter theory implies that the Riccati recursion converges to the unique stabilising solution $\Sigma_\infty\succeq0$. 
Hence $\sup_t \mathbb E\|e_t\|^2<\infty$ and $\mathbb E\|e_t\|^2\to \mathrm{tr}\,\Sigma_\infty$.

\emph{Step 2: Exponential convergence under hard anchoring.}  
If $Q=0$ and $\sigma_\star^2=0$ for all $t\ge T_F$, then beyond $T_F$ the error recursion reduces to
\[
e_{t+1} = (I-K_t H_{\mathrm{aug}})A e_t .
\]
Detectability ensures convergence of $K_t$ to a stabilising gain $K_\star$ such that $(I-K_\star H_{\mathrm{aug}})A$ is Schur. 
Uniform exponential stability of the time--varying system follows, implying $e_t\to 0$ exponentially.

\emph{Step 3: Transfer to PD convergence.}  
By local Lipschitz continuity of $\phi$ in $m$ near $M^\star$,
\[
\big|\,\phi(\pi_t,\hat M_{t|t})-\phi(\pi_t,M^\star)\,\big|
\;\le\; L\,\|e_t\|.
\]
Since $e_t\to0$ exponentially, the right--hand side vanishes in probability. 
Thus the lifetime PD sequence converges in probability to the neutral macro limit. 
\end{proof}

\begin{remark}[Continuous anchoring within the forecast window]
If $Q\succeq 0$ and $\sigma_\star^2>0$ remain constant, then $e_t$ converges in distribution to a Gaussian with covariance $\Sigma_\infty$. 
The induced PD deviations satisfy 
\[
\limsup_{t\to\infty}\mathbb E\big|\,\phi(\pi_t,\hat M_{t|t})-\phi(\pi_t,M^\star)\,\big| 
\;\le\; L\,\mathbb E\|Z\|, \qquad Z\sim\mathcal N(0,\Sigma_\infty).
\]
This yields mean--square boundedness and practical stability of the PD term structure. 
Exact convergence is obtained when $Q$ and $\sigma_\star^2$ are reduced after $T_F$.
\end{remark}

\medskip
\noindent
The preceding theorem establishes stability for the generalised anchored (stacked) model. 
For completeness, we also record the parallel result for the simpler anchored deviation 
formulation introduced in Section~\ref{sec:errordyn}. 

\begin{corollary}[Anchored deviation model: stability and convergence]\label{cor:anch_dev}
Let $\xi_t := M_t - M^\star$ with dynamics
\[
\xi_{t+1}=A\xi_t+w_t,\qquad w_t\sim\mathcal N(0,Q),
\]
and observation
\[
y_t=H\xi_t+\nu_t,\qquad \nu_t\sim\mathcal N(0,R).
\]
Assume $(A,Q^{1/2})$ stabilisable and $(A,H)$ detectable. 
Let $\hat\xi_{t|t}$ be the Kalman estimate, $e_t=\hat\xi_{t|t}-\xi_t$, and evaluate the PIT matrix at $M^\star+\hat\xi_{t|t}$. 
If $\phi(\pi,m)$ is locally Lipschitz in $m$ near $M^\star$ with constant $L>0$, then:
\begin{enumerate}
\item (\emph{Mean--square boundedness}) $\Sigma_t\to\Sigma_\infty\succeq0$ and $\lim_{t\to\infty}\mathbb E\|e_t\|^2=\mathrm{tr}\,\Sigma_\infty<\infty$.
\item (\emph{Exponential convergence under hard anchoring}) If $Q=0$ and $R=0$ for all $t\ge T_F$, then $e_t\to 0$ exponentially and
\[
\phi(\pi_t,M^\star+\hat\xi_{t|t})\;\xrightarrow{\;p\;}\;\phi(\pi_t,M^\star).
\]
\end{enumerate}
\end{corollary}

\begin{proof}
Standard Riccati theory gives mean--square boundedness. 
If $Q=0$ and $R=0$ beyond $T_F$, then $e_{t+1}=(I-K_t H)A e_t$ with $K_t\to K_\star$ and $(I-K_\star H)A$ Schur, 
so $e_t\to0$ exponentially. 
Lipschitz continuity of $\phi$ yields the PD convergence. 
\end{proof}

\bibliographystyle{IEEEtran}
\phantomsection  
\bibliography{referencesUpdated}

\begin{thebibliography}{10}
\providecommand{\url}[1]{#1}
\csname url@samestyle\endcsname
\providecommand{\newblock}{\relax}
\providecommand{\bibinfo}[2]{#2}
\providecommand{\BIBentrySTDinterwordspacing}{\spaceskip=0pt\relax}
\providecommand{\BIBentryALTinterwordstretchfactor}{4}
\providecommand{\BIBentryALTinterwordspacing}{\spaceskip=\fontdimen2\font plus
\BIBentryALTinterwordstretchfactor\fontdimen3\font minus \fontdimen4\font\relax}
\providecommand{\BIBforeignlanguage}[2]{{%
\expandafter\ifx\csname l@#1\endcsname\relax
\typeout{** WARNING: IEEEtran.bst: No hyphenation pattern has been}%
\typeout{** loaded for the language `#1'. Using the pattern for}%
\typeout{** the default language instead.}%
\else
\language=\csname l@#1\endcsname
\fi
#2}}
\providecommand{\BIBdecl}{\relax}
\BIBdecl

\bibitem{IASB2014}
\BIBentryALTinterwordspacing
I.~A.~S. Board, ``{IFRS 9} financial instruments,'' 2014, london. [Online]. Available: \url{https://www.ifrs.org/issued-standards/list-of-standards/ifrs-9-financial-instruments/}
\BIBentrySTDinterwordspacing

\bibitem{FASB2016}
\BIBentryALTinterwordspacing
F.~A.~S. Board, ``Financial instruments—credit losses (topic 326)—purchased financial assets,'' 2016, norwalk, CT. [Online]. Available: \url{https://www.fasb.org/projects/current-projects/financial-instruments%E2%80%94credit-losses-(topic-326)%E2%80%94purchased-financial-assets-401651}
\BIBentrySTDinterwordspacing

\bibitem{Basel2000}
\BIBentryALTinterwordspacing
B.~C. on~Banking~Supervision, ``Principles for the management of credit risk,'' 2000, basel, Switzerland. [Online]. Available: \url{https://www.bis.org/publ/bcbs75.pdf}
\BIBentrySTDinterwordspacing

\bibitem{Tasche2019}
\BIBentryALTinterwordspacing
D.~Tasche, ``The art of probability-of-default curve calibration,'' \emph{arXiv preprint arXiv:1212.3716}, 2012. [Online]. Available: \url{https://arxiv.org/pdf/1212.3716}
\BIBentrySTDinterwordspacing

\bibitem{ECB2017}
\BIBentryALTinterwordspacing
E.~C. Bank, ``Guidance to banks on non-performing loans – annex: Expected credit losses,'' 2017, frankfurt. [Online]. Available: \url{https://www.bankingsupervision.europa.eu/ecb/pub/pdf/guidance_on_npl.en.pdf}
\BIBentrySTDinterwordspacing

\bibitem{EBA2017}
\BIBentryALTinterwordspacing
E.~B. Authority, ``Eba reports on results of the second impact assessment of ifrs 9,'' 2017, london. [Online]. Available: \url{https://www.iasplus.com/en/news/2017/07/eba-ifrs-9}
\BIBentrySTDinterwordspacing

\bibitem{BCBS2015}
\BIBentryALTinterwordspacing
B.~C. on~Banking~Supervision, ``Guidance on credit risk and accounting for expected credit losses,'' 2015, basel, Switzerland. [Online]. Available: \url{https://www.bis.org/bcbs/publ/d350.pdf}
\BIBentrySTDinterwordspacing

\bibitem{CroushoreStark2001}
\BIBentryALTinterwordspacing
D.~Croushore and T.~Stark, ``A real-time data set for macroeconomists,'' \emph{Journal of Econometrics}, vol. 105, no.~1, pp. 111--130, 2001. [Online]. Available: \url{https://doi.org/10.1016/S0304-4076(01)00072-0}
\BIBentrySTDinterwordspacing

\bibitem{ReifschneiderTulip2017}
\BIBentryALTinterwordspacing
D.~Reifschneider and P.~Tulip, ``Gauging the uncertainty of the economic outlook,'' \emph{International Journal of Forecasting}, vol.~35, no.~4, pp. 1564--1582, 2019. [Online]. Available: \url{https://doi.org/10.1016/j.ijforecast.2018.07.016}
\BIBentrySTDinterwordspacing

\bibitem{OECD2020}
\BIBentryALTinterwordspacing
O.~for Economic Co-operation and Development, ``Oecd economic outlook, june 2020,'' 2020, paris. [Online]. Available: \url{https://doi.org/10.1787/0d1d1e2e-en}
\BIBentrySTDinterwordspacing

\bibitem{IMF2021}
\BIBentryALTinterwordspacing
I.~M. Fund, ``World economic outlook: Recovery during a pandemic,'' 2021, washington, DC. [Online]. Available: \url{https://www.imf.org/en/Publications/WEO/Issues/2021/10/12/world-economic-outlook-october-2021}
\BIBentrySTDinterwordspacing

\bibitem{EBA2020}
\BIBentryALTinterwordspacing
E.~B. Authority, ``Eba report on the implementation of ifrs~9 by eu institutions,'' European Banking Authority, Tech. Rep., 2023. [Online]. Available: \url{https://www.eba.europa.eu/sites/default/files/2023-11/25b12d35-9c28-4335-a589-166c77198920/Final%20Report%20on%20IFRS9%20implementation%20by%20EU%20institutions.pdf}
\BIBentrySTDinterwordspacing

\bibitem{BIS2021}
\BIBentryALTinterwordspacing
B.~for International~Settlements, ``Credit risk modelling current practices and applications,'' BIS, Tech. Rep., 1999. [Online]. Available: \url{https://www.bis.org/publ/bcbs49.pdf}
\BIBentrySTDinterwordspacing

\bibitem{Borio2014}
\BIBentryALTinterwordspacing
C.~Borio, M.~Drehmann, and K.~Tsatsaronis, ``Stress-testing macro stress testing: Does it live up to expectations?'' \emph{Journal of Financial Stability}, vol.~12, pp. 3--15, 2014. [Online]. Available: \url{https://doi.org/10.1016/j.jfs.2013.06.001}
\BIBentrySTDinterwordspacing

\bibitem{Orphanides2002}
\BIBentryALTinterwordspacing
A.~Orphanides and S.~van Norden, ``The unreliability of output-gap estimates in real time,'' \emph{Review of Economics and Statistics}, vol.~84, no.~4, pp. 569--583, 2002. [Online]. Available: \url{https://doi.org/10.1162/003465302760556422}
\BIBentrySTDinterwordspacing

\bibitem{Hamilton1994}
\BIBentryALTinterwordspacing
J.~D. Hamilton, \emph{Time Series Analysis}.\hskip 1em plus 0.5em minus 0.4em\relax Princeton, NJ: Princeton University Press, 1994. [Online]. Available: \url{https://doi.org/10.1515/9780691218632}
\BIBentrySTDinterwordspacing

\bibitem{Kreinin2001}
A.~Kreinin and M.~Sidelnikova, ``Regularization algorithms for transition matrices,'' \emph{Algo Research Quarterly}, vol.~4, no. 1/2, pp. 23--40, 2001.

\bibitem{JarrowLandoTurnbull1997}
\BIBentryALTinterwordspacing
R.~A. Jarrow, D.~Lando, and S.~M. Turnbull, ``A markov model for the term structure of credit risk spreads,'' \emph{Review of Financial Studies}, vol.~10, no.~2, pp. 481--523, 1997. [Online]. Available: \url{https://doi.org/10.1093/rfs/10.2.481}
\BIBentrySTDinterwordspacing

\bibitem{DuffieSingleton1999}
\BIBentryALTinterwordspacing
D.~Duffie and K.~J. Singleton, ``Modeling term structures of defaultable bonds,'' \emph{Review of Financial Studies}, vol.~12, no.~4, pp. 687--720, 1999. [Online]. Available: \url{https://doi.org/10.1093/rfs/12.4.687}
\BIBentrySTDinterwordspacing

\bibitem{Lando2004}
\BIBentryALTinterwordspacing
D.~Lando, \emph{Credit Risk Modeling: Theory and Applications}.\hskip 1em plus 0.5em minus 0.4em\relax Princeton, NJ: Princeton University Press, 2004. [Online]. Available: \url{https://doi.org/10.1515/9781400829194}
\BIBentrySTDinterwordspacing

\bibitem{Kavvathas2001}
\BIBentryALTinterwordspacing
D.~Kavvathas, ``Estimating credit rating transition probabilities for corporate bonds,'' \emph{AFA}, 2001. [Online]. Available: \url{https://papers.ssrn.com/sol3/papers.cfm?abstract_id=252517}
\BIBentrySTDinterwordspacing

\bibitem{Christensen2004}
\BIBentryALTinterwordspacing
J.~H.~E. Christensen, E.~B. Hansen, and D.~Lando, ``Confidence sets for continuous-time rating transition matrices,'' \emph{Journal of Banking \& Finance}, vol.~28, no.~11, pp. 2575--2602, 2004. [Online]. Available: \url{https://doi.org/10.1016/j.jbankfin.2004.06.003}
\BIBentrySTDinterwordspacing

\bibitem{Kalman1960}
\BIBentryALTinterwordspacing
R.~E. Kalman, ``A new approach to linear filtering and prediction problems,'' \emph{Transactions of the ASME--Journal of Basic Engineering}, vol.~82, no.~1, pp. 35--45, 1960. [Online]. Available: \url{https://doi.org/10.1115/1.3662552}
\BIBentrySTDinterwordspacing

\bibitem{Maybeck1979}
\BIBentryALTinterwordspacing
P.~S. Maybeck, \emph{Stochastic Models, Estimation and Control, Vol. 1}.\hskip 1em plus 0.5em minus 0.4em\relax Academic Press, 1979. [Online]. Available: \url{https://www.cs.unc.edu/~welch/kalman/media/pdf/maybeck_ch1.pdf}
\BIBentrySTDinterwordspacing

\bibitem{Simon2006}
\BIBentryALTinterwordspacing
D.~Simon, \emph{Optimal State Estimation: Kalman, H Infinity, and Nonlinear Approaches}.\hskip 1em plus 0.5em minus 0.4em\relax New York: Wiley, 2006. [Online]. Available: \url{https://doi.org/10.1002/0470045345}
\BIBentrySTDinterwordspacing

\bibitem{Nickell2000}
\BIBentryALTinterwordspacing
P.~Nickell, W.~Perraudin, and S.~Varotto, ``Stability of rating transitions,'' \emph{Journal of Banking \& Finance}, vol.~24, no. 1-2, pp. 203--227, 2000. [Online]. Available: \url{https://doi.org/10.1016/S0378-4266(99)00057-6}
\BIBentrySTDinterwordspacing

\bibitem{LandoSkodeberg2002}
\BIBentryALTinterwordspacing
D.~Lando and T.~M. Sk{\o}deberg, ``Analyzing rating transitions and rating drift with continuous observations,'' \emph{Journal of Banking \& Finance}, vol.~26, no. 2-3, pp. 423--444, 2002. [Online]. Available: \url{https://doi.org/10.1016/S0378-4266(01)00228-X}
\BIBentrySTDinterwordspacing

\bibitem{Israel2001}
\BIBentryALTinterwordspacing
R.~B. Israel, J.~S. Rosenthal, and J.~Z. Wei, ``Finding generators for markov chains via empirical transition matrices, with applications to credit ratings,'' \emph{Mathematical Finance}, vol.~11, no.~2, pp. 245--265, 2001. [Online]. Available: \url{https://doi.org/10.1111/1467-9965.00114}
\BIBentrySTDinterwordspacing

\bibitem{Bluhm2016}
\BIBentryALTinterwordspacing
C.~Bluhm, L.~Overbeck, and C.~Wagner, \emph{Introduction to Credit Risk Modeling}, 2nd~ed.\hskip 1em plus 0.5em minus 0.4em\relax Chapman \& Hall/CRC, 2016. [Online]. Available: \url{https://doi.org/10.1201/9781584889939}
\BIBentrySTDinterwordspacing

\bibitem{Bangia2002}
\BIBentryALTinterwordspacing
A.~Bangia, F.~X. Diebold, A.~Kronimus, C.~Schagen, and T.~Schuermann, ``Ratings migration and the business cycle, with applications to credit portfolio stress testing,'' \emph{Journal of Banking \& Finance}, vol.~26, no. 2-3, pp. 445--474, 2002. [Online]. Available: \url{https://doi.org/10.1016/S0378-4266(01)00229-1}
\BIBentrySTDinterwordspacing

\bibitem{Pesaran2006}
\BIBentryALTinterwordspacing
M.~H. Pesaran, T.~Schuermann, B.-J. Treutler, and S.~M. Weiner, ``Macroeconomic dynamics and credit risk: a global perspective,'' \emph{Journal of Money, Credit and Banking}, pp. 1211--1261, 2006. [Online]. Available: \url{: https://www.jstor.org/stable/3839005}
\BIBentrySTDinterwordspacing

\bibitem{Duffie2007}
\BIBentryALTinterwordspacing
D.~Duffie, L.~Saita, and K.~Wang, ``Multi-period corporate default prediction with stochastic covariates,'' \emph{Journal of Financial Economics}, vol.~83, no.~3, pp. 635--665, 2007. [Online]. Available: \url{https://doi.org/10.1016/j.jfineco.2005.10.011}
\BIBentrySTDinterwordspacing

\bibitem{BCBS2005}
\BIBentryALTinterwordspacing
B.~C. on~Banking~Supervision, ``Studies on the validation of internal rating systems,'' Basel Committee on Banking Supervision, Tech. Rep.~14, 2005. [Online]. Available: \url{https://www.bis.org/publ/bcbs_wp14.pdf}
\BIBentrySTDinterwordspacing

\bibitem{NelsonPlosser1982}
\BIBentryALTinterwordspacing
C.~R. Nelson and C.~I. Plosser, ``Trends and random walks in macroeconomic time series: Some evidence and implications,'' \emph{Journal of Monetary Economics}, vol.~10, no.~2, pp. 139--162, 1982. [Online]. Available: \url{https://doi.org/10.1016/0304-3932(82)90012-5}
\BIBentrySTDinterwordspacing

\bibitem{Vasicek1977}
\BIBentryALTinterwordspacing
O.~Vasicek, ``An equilibrium characterization of the term structure,'' \emph{Journal of Financial Economics}, vol.~5, no.~2, pp. 177--188, 1977. [Online]. Available: \url{https://doi.org/10.1016/0304-405X(77)90016-2}
\BIBentrySTDinterwordspacing

\bibitem{Seneta2006}
\BIBentryALTinterwordspacing
E.~Seneta, \emph{Non-Negative Matrices and Markov Chains}, ser. Springer Series in Statistics.\hskip 1em plus 0.5em minus 0.4em\relax Springer, 2006, revised edition. [Online]. Available: \url{https://doi.org/10.1007/0-387-32792-4}
\BIBentrySTDinterwordspacing

\bibitem{Mitrophanov2005}
\BIBentryALTinterwordspacing
A.~Y. Mitrophanov, ``Sensitivity and convergence of uniformly ergodic markov chains,'' \emph{Journal of Applied Probability}, vol.~42, no.~4, pp. 1003--1014, 2005. [Online]. Available: \url{https://doi.org/10.1239/jap/1134587812}
\BIBentrySTDinterwordspacing

\bibitem{KushnerYin2003}
\BIBentryALTinterwordspacing
H.~J. Kushner and G.~G. Yin, \emph{Stochastic approximation and recursive algorithms and applications}.\hskip 1em plus 0.5em minus 0.4em\relax Springer, 2003. [Online]. Available: \url{https://nzdr.ru/data/media/biblio/kolxoz/M/MN/MNs/Kushner%20H.,%20Yin%20G.%20Stochastic%20Approximation%20and%20Recursive%20Algorithms%20and%20Applications%20(Springer,%202003)(497s)_MNs_.pdf}
\BIBentrySTDinterwordspacing

\bibitem{Doucet2001}
\BIBentryALTinterwordspacing
A.~Doucet, N.~de~Freitas, and N.~Gordon, Eds., \emph{Sequential Monte Carlo Methods in Practice}.\hskip 1em plus 0.5em minus 0.4em\relax New York: Springer, 2001. [Online]. Available: \url{https://doi.org/10.1007/978-1-4757-3437-9}
\BIBentrySTDinterwordspacing

\bibitem{Durbin2012}
\BIBentryALTinterwordspacing
J.~Durbin and S.~J. Koopman, \emph{Time Series Analysis by State Space Methods}, 2nd~ed.\hskip 1em plus 0.5em minus 0.4em\relax Oxford University Press, 2012. [Online]. Available: \url{https://doi.org/10.1093/acprof:oso/9780199641178.001.0001}
\BIBentrySTDinterwordspacing

\bibitem{Basel2018Stress}
\BIBentryALTinterwordspacing
B.~C. on~Banking~Supervision, ``Stress testing principles,'' Basel Committee on Banking Supervision, Tech. Rep., October 2018. [Online]. Available: \url{https://www.bis.org/bcbs/publ/d450.pdf}
\BIBentrySTDinterwordspacing

\bibitem{NGFS2020}
\BIBentryALTinterwordspacing
N.~C.~S. for Central~Banks and Supervisors, ``Ngfs climate scenarios for central banks and supervisors,'' Network for Greening the Financial System (NGFS), Tech. Rep., June 2020. [Online]. Available: \url{https://www.ngfs.net/en/publications-and-statistics/publications/ngfs-climate-scenarios-central-banks-and-supervisors}
\BIBentrySTDinterwordspacing

\bibitem{Tsay2010}
\BIBentryALTinterwordspacing
R.~S. Tsay, \emph{Analysis of Financial Time Series}, 3rd~ed.\hskip 1em plus 0.5em minus 0.4em\relax Hoboken, NJ: Wiley, 2010. [Online]. Available: \url{https://doi.org/10.1002/9780470644560}
\BIBentrySTDinterwordspacing

\end{thebibliography}

\end{document}